\documentclass[journal]{IEEEtran}
\usepackage{algorithm,comment}

\usepackage{array,subfigure}
\usepackage{stfloats}
\usepackage{url}
\usepackage{verbatim}
\usepackage{cite}
\usepackage{amsmath,amssymb,amsfonts}
\usepackage{graphicx}
\usepackage{textcomp}
\usepackage[compact]{titlesec}
\usepackage[dvipsnames]{xcolor}
\usepackage[nocomma]{optidef}
\usepackage[acronyms,nonumberlist,nopostdot,nomain,nogroupskip]{glossaries}
\usepackage{algorithm}
\usepackage{algpseudocode}
\usepackage{amsthm}
\usepackage{url}
\usepackage[caption=false,font=normalsize,labelfont=sf,textfont=sf]{subfig}
\def\BibTeX{{\rm B\kern-.05em{\sc i\kern-.025em b}\kern-.08em
    T\kern-.1667em\lower.7ex\hbox{E}\kern-.125emX}}
    
\usepackage[many]{tcolorbox}
\usepackage{lipsum}

\renewenvironment{IEEEbiography}[1]
  {\IEEEbiographynophoto{#1}}
  {\endIEEEbiographynophoto}

\makeatletter
\def\@IEEEbiography{\vspace{-2\baselineskip}\IEEEbiography}
\makeatother

\definecolor{titlebg}{RGB}{100,22,72}
\definecolor{introbg}{RGB}{0,128,128}

\usepackage[textwidth=2.5cm,backgroundcolor=yellow]{todonotes}

\newcommand{\edits}[1]{\textcolor{black}{#1}}

\newcommand{\revision}[1]{\textcolor{black}{#1}}
\newcommand{\riccardoandfrancesco}[1]{\textcolor{green!60!black}{#1}}

\newcommand{\wb}{\mathbf{w}}
\newcommand{\Nc}{\mathcal{N}}

\newtcolorbox{usecase}[1][]{
  breakable,
  enlarge top by=0.05in,
  enhanced,
  arc=0pt,
  outer arc=0pt,
  colframe=titlebg,
  colback=titlebg!05,
  overlay unbroken and first={
    \node[
      draw=titlebg,
      fill=titlebg,
      rotate=0,
      anchor=north west,
      text=white,
      font=\bfseries
    ]
    at (frame.north west)  
    {#1};
  }
}

\newtcolorbox{mission}[1][]{
  breakable,
  enhanced,
  arc=0pt,
  outer arc=0pt,
  colframe=introbg,
  colback=introbg!05,
  overlay unbroken and first={
    \node[
      draw=introbg,
      fill=introbg,
      rotate=0,
      anchor=north west,
      text=white,
      font=\bfseries
    ]
    at (frame.north west)  
    {#1};
  }
}

\usepackage{xspace}

\newacronym{twt}{TWT}{Target Wake Time}
\newacronym{ofdma}{OFDMA}{Orthogonal Frequency Division Multiple Access}
\newacronym{mimo}{MIMO}{multiple-input multiple-output}
\newacronym{mumimo}{MU-MIMO}{Multi-user MIMO}
\newacronym{tsn}{TSN}{Time Sensitive Networking}
\newacronym{bss}{BSS}{basic service set}
\newacronym{ap}{AP}{Access Point}
\newacronym{sta}{STA}{station}
\newacronym{ru}{RU}{Resource Unit}
\newacronym{cbw}{CBW}{Channel Bandwidth}
\newacronym{he}{HE}{High Efficiency}
\newacronym{ul}{UL}{uplink}
\newacronym{dl}{DL}{downlink}
\newacronym{sp}{SP}{Service Period}
\newacronym{mcs}{MCS}{Modulation and Coding Scheme}
\newacronym{aoi}{AoI}{Age of Information}
\newtheorem{theorem}{Theorem}

\makeatletter
\def\ps@IEEEtitlepagestyle{%
  \def\@oddfoot{\mycopyrightnotice}%
  \def\@oddhead{\hbox{}\@IEEEheaderstyle\leftmark\hfil\thepage}\relax
  \def\@evenhead{\@IEEEheaderstyle\thepage\hfil\leftmark\hbox{}}\relax
  \def\@evenfoot{}%
}
\def\mycopyrightnotice{%
  \begin{minipage}{\textwidth}
  \centering \scriptsize
  \copyright~2025 IEEE.  Personal use of this material is permitted.  Permission from IEEE must be obtained for all other uses, in any current or future media, including reprinting/republishing this material for advertising or promotional purposes, creating new collective works, for resale or redistribution to servers or lists, or reuse of any copyrighted component of this work in other works.
  \end{minipage}
}
\makeatother

\begin{document}

\title{Target Wake Time Scheduling for Time-sensitive  and Energy-efficient  Wi-Fi Networks

\author{Fabio Busacca,~\IEEEmembership{Member,~IEEE,}
        Corrado Puligheddu,~\IEEEmembership{Member,~IEEE,}
        Francesco Raviglione,~\IEEEmembership{Member,~IEEE,}
        Riccardo Rusca,~\IEEEmembership{Member,~IEEE,}
        Claudio Casetti,~\IEEEmembership{Senior Member,~IEEE,}
        Carla Fabiana Chiasserini,~\IEEEmembership{Fellow,~IEEE,}
        and~Sergio Palazzo,~\IEEEmembership{Senior Member,~IEEE}}
\thanks{F. Busacca and S. Palazzo are with University of Catania, Italy.}
\thanks{C. Casetti, C. F. Chiasserini, C. Puligheddu, F. Raviglione, and R. Rusca  are with
Politecnico di Torino, Italy.}

\thanks{This work was supported by the European Commission through Grant No. 101095890 (Predict-6G project), and by the EU under the Italian NRRP of NextGenerationEU, through the RESTART program (PE0000001) and the MOST CNMS  (CN00000023). This manuscript reflects only the authors’ views and opinions, neither the EU nor the EC can be considered responsible for them.\\
This is an extended version of our IEEE PIRMC 2024 paper~\cite{puligheddu2024target}.}
}

\maketitle

\begin{abstract}
Time Sensitive Networking (TSN) is fundamental for the reliable, low-latency networks that will enable the Industrial Internet of Things (IIoT). Wi-Fi has historically been considered unfit for TSN, as channel contention and collisions prevent deterministic transmission delays. However, this issue can be overcome by using Target Wake Time (TWT), which 
enables the access point to instruct Wi-Fi stations to wake up and transmit in non-overlapping TWT Service Periods (SPs), and sleep in the remaining time.
In this paper, we first formulate the TWT Acceptance and Scheduling Problem (TASP), with the objective  to schedule TWT SPs that maximize traffic throughput and energy efficiency while respecting Age of Information (AoI) constraints. Then, due to TASP being NP-hard, we propose the TASP Efficient Resolver (TASPER), a heuristic strategy to find near-optimal solutions efficiently. {\color{black}Using a TWT  simulator based on ns-3, we compare TASPER to several baselines, including HSA, a state-of-the-art solution originally designed for WirelessHART networks.} We 
demonstrate that TASPER obtains up to 24.97\% lower mean transmission rejection cost and saves up to 14.86\% more energy compared to the leading baseline, ShortestFirst, in a challenging, large-scale scenario. \revision{Additionally, when compared to HSA, TASPER also reduces the energy consumption by 34\% and reduces the mean rejection cost by 26\%. }
Furthermore, we validate TASPER on our IIoT testbed, which comprises 10 commercial TWT-compatible stations, observing that our solution admits more transmissions than the best baseline strategy, without violating any AoI deadline.
\end{abstract}


\section{Introduction}

The transformation brought about by technologies like the Industrial Internet of Things (IIoT), Artificial Intelligence (AI), and cyber-physical systems is changing the way industries and production chains operate~\cite{munirathinam2020industry}. In this new industrial paradigm, the traditional boundaries of manufacturing are redefined through smart factories, where real-time data, interconnected systems, and automation are seamlessly integrated. With IIoT at its core, this transformation ushers in a wave of intelligent sensors, actuators, and data-driven decision systems that promise unprecedented efficiency and precision in industrial processes.

Reliable, low-latency communication is essential in smart factories 
{\color{black} where deployments may involve even hundreds of IIoT devices exchanging time-sensitive data to monitor production lines, detect anomalies, and coordinate actions across complex processes.}
As a result, communication delays and network reliability become critical, not only to maintain efficiency but also to avoid costly interruptions that could lead to production downtime, equipment damage, or even safety hazards. However, meeting these communication demands is challenging in a wireless environment, where signals are subject to interference and collisions that can delay critical messages.

While ARQ (Automatic Repeat reQuest)  mechanisms can sometimes recover lost transmissions, they cannot ensure that data remains relevant. A key challenge is managing the \gls{aoi}, i.e., the time elapsed between data generation and its reception at the destination. High AoI can make data obsolete, especially in cases where an alert about a production anomaly arrives too late for effective intervention. For example, delayed sensor data could leave operators with no choice but to halt production, leading to significant financial and operational costs.

To overcome these challenges, traditional approaches might suggest isolating each device on a dedicated wireless channel, but this is rarely feasible given the limited spectrum and high device density in IIoT settings. Thus, a more sophisticated approach is needed to manage channel access while minimizing latency and ensuring data freshness.

\gls{twt}, a feature introduced with \mbox{Wi-Fi}~6 (IEEE 802.11ax), offers a promising solution for managing time-sensitive communication in the IIoT. TWT allows an access point (AP) to pre-define wake-up times for each \gls{sta} in the network, scheduling them to transmit and receive data in dedicated slots and enabling them to sleep in the meantime. This level of control provides several critical advantages: 

\begin{itemize}
\item Enhanced energy efficiency: Many IIoT devices run on battery power, and frequent data exchange can quickly drain their energy. With TWT, devices can remain in a low-power mode and only wake up for scheduled transmissions, significantly extending their operational life.
\item Reduced channel contention: By allocating exclusive transmission times for each device, TWT helps  prevent contention, reducing the chances of collisions and network delays. This allows for more deterministic communication, essential for safety-critical and real-time applications in manufacturing.
\item Improved predictability and network stability: TWT scheduling makes it possible to better anticipate network load and reduce jitter, which can be vital for industrial processes that depend on synchronized device interactions and precise timing.
\end{itemize}

While TWT can address the issues of timing and energy-efficient operation in IIoT, it leaves open the question of \textit{how to strategically assign wake times to devices} to best balance data freshness, energy savings, and throughput. Without a systematic scheduling approach, there is no guarantee that TWT allocations will meet the demanding latency and reliability requirements of time-sensitive applications. Therefore, a robust TWT scheduling mechanism becomes essential for IIoT networks that use IEEE 802.11ax. Such a mechanism should:

\begin{itemize}
\item[{\em (i)}] Ensure that devices wake up only when necessary, reducing energy consumption;
\item[{\em (ii)}] Maximize the admission rate of critical and timely data flows to prevent stale or outdated information;
\item[{\em (iii)}] Minimize AoI, ensuring that time-sensitive information remains fresh and actionable.
\end{itemize}

\revision{In this work, we thus focus  on making the TWT mechanism suitable to the stringent requirements of IIoT traffic by scheduling device wake times to: optimize energy saving,  accommodate as many traffic flows as possible while prioritizing critical data, and  guarantee a timely delivery of information.}
Hence, with this goal in mind,  we introduce the TWT Acceptance and Scheduling Problem Efficient Resolver (TASPER), an algorithmic solution for energy-efficient, AoI-constrained TWT scheduling in Wi-Fi networks. 

Our contributions can be summarized as follows:
    
     $\bullet$\, We model a  TWT-enabled  Wi-Fi networks, where \glspl{sta} generate AoI-constrained traffic. Such a model captures the relationship between the traffic generated and transmitted by the Wi-Fi \glspl{sta}, the \gls{sta} power states, and the associated energy consumption.
    
     $\bullet$\, We build upon this model to define the TWT Acceptance and Scheduling Problem (TASP), whose objective is to minimize the rejected (therefore, not scheduled) traffic and the stations' energy consumption while  satisfying the maximum AoI constraint. Finding a solution to the TASP implies determining  \textit{(i)} whether or not to admit the traffic, and \textit{(ii)} the target wake time(s) of each STA in the network. Notably,  the possibility to reject traffic may allow for a schedule with sufficient room for higher-priority traffic.
    
     $\bullet$\, As the TASP is NP-hard, to solve it efficiently we devise a novel heuristic strategy \revision{that efficiently produces high-quality scheduling decisions while accounting for energy consumption. The algorithm is designed to be lightweight and fast, making it suitable for timely execution in dynamic IIoT environments}.
     
    $\bullet$\, We enhance the ns-3-based simulator for TWT by Venkateswaran et al.~\cite{ns-3-twt-shyam-2024-accepted-paper}, introducing several novel features to simulate advanced TWT scheduling approaches such as TASPER. The simulator, that we call \emph{ns-3-twt}, includes the possibility to simulate deadlines, different classes of devices and multiple TWT networks.
    
    $\bullet$\, We thoroughly evaluate TASPER and other baselines, leveraging ns-3-twt and show that it offers substantial improvements over the other strategies. Specifically, TASPER obtains up to 24.97\% lower mean rejection cost and saves up to 14.86\% more energy compared to the best baseline. {\color{black} Additionally, it reduces  energy consumption up to 34\% and the mean rejection cost of up to 26\% when compared to HSA, a state-of-the-art approach originally designed for WirelessHART networks.}

    $\bullet$\, We design and implement an IIoT testbed comprising 10 commercial TWT-compatible STAs. Through our testbed, we validate TWT as a means for energy-efficient traffic scheduling without incurring channel contention delays, observing 49\% energy saving and 16\% lower median AoI. Also, we remark that TASPER admits more transmissions than the best baseline strategy, confirming  its effectiveness in scheduling transmissions.

The rest of the paper is organized as follows. Sec.\,\ref{sec:twt_basics} {\color{black}highlights the advantages of the TWT feature and motivates the need for a transmission scheduling that accounts for traffic generation times and deadlines. Sec.\,\ref{sec:system_model} introduces our  system model and the TWT scheduling problem, which is shown to be NP-hard. 
}
Sec.\,\ref{sec:TASPER_algo} presents the key ideas and principles behind the TASPER scheduling algorithm we propose. After describing in Sec.\,\ref{sec:methodology} our ns-3-based simulation framework and evaluation methodology, and the scheduling baselines considered as  benchmarks for TASPER,  Sec.\,\ref{sec:peva} shows the numerical results we obtained. Sec.\,\ref{sec:exp}  introduces the experimental testbed we set up with commercial off-the-shelf devices and  validates the feasibility and good performance of our solution. Finally, Sec.\,\ref{sec:related_work} discusses some relevant related work, and Sec.\,\ref{sec:conclusions} draws our conclusions.

{\color{black}
\section{TWT Operations: The Need for a Scheduling Strategy}
\label{sec:twt_basics}

TWT enables the \mbox{Wi-Fi} \glspl{sta} to negotiate awake periods with the \gls{ap} to exchange traffic and to enter a doze mode during the remaining time to save energy~\cite{802.11ax}. 
One of the main advantages of \gls{twt} is therefore energy efficiency. 
}
Indeed, as the STAs are able to enter a sleep mode while not involved in the transmission or reception of data, they can save a substantial amount of energy that would instead be spent turning the radio into receive  mode for data that is not destined to them, or in other operations such as Clear Channel Assessment  
\cite{802.11ax,ns-3-dev-docs}.
The advantage of using TWT in terms of energy savings, which is especially important in energy-constrained battery-powered IIoT sensors, is highlighted in Fig.~\ref{fig:stateplot}.
The plots show the output of  our ns-3-twt simulator (described in Sec.~\ref{sub:sim_ns-3-twt}) in a scenario with 8 \glspl{sta}, that need to send their traffic to an AP, modeling IIoT sensors of different kinds and with different energy figures. The traffic consists of a 1500-byte UDP packet encoding sensor data to be transmitted to the AP. Each STA $m$ ($m{=}1,\ldots,8$) receives a UDP packet from its higher layers at a time equal to $(m\cdot 5)$\,ms. We also consider that the transmission periodicity of the STAs is known (which is likely to happen for common IIoT sensors that transmit their data periodically); thus, the STAs can conveniently employ a \gls{twt} mechanism in which each \gls{twt} SP starts as soon as a UDP packet is received at a STA from its higher layers. Each TWT SP is set to last slightly more than 5\,ms, just for the sake of showing a clear illustrative example.
Fig.~\ref{fig:stateplot} shows the difference with (a) and without (b) \gls{twt}, depicting the evolution of the physical layer states of STA\,2 over time, and including the total energy consumption in a reference 32-ms period. One can observe that, with \gls{twt}, the \gls{sta} spends most of the time in the very low energy  sleep state, except for the time in which it has to receive the beacon from the AP (the receive  mode right after the beginning of the simulation) and during its \gls{sp}, in which it transmits its data. Conversely, without TWT, the STA never enters the \textit{sleep} state, and alternates between the \textit{idle} and the \textit{receive} state.  Specifically, such states are entered every time the STA is overhearing, i.e.,  another STA transmits  data  to the AP. Since the idle state is associated with a higher  energy consumption than the sleep state,  
employing \gls{twt} results in a significant $3.4\times$ energy saving  over a reference period of just 32\,ms.

\begin{figure}[t]
\centering
\subfigure[With TWT]{\includegraphics[width=0.35\textwidth]{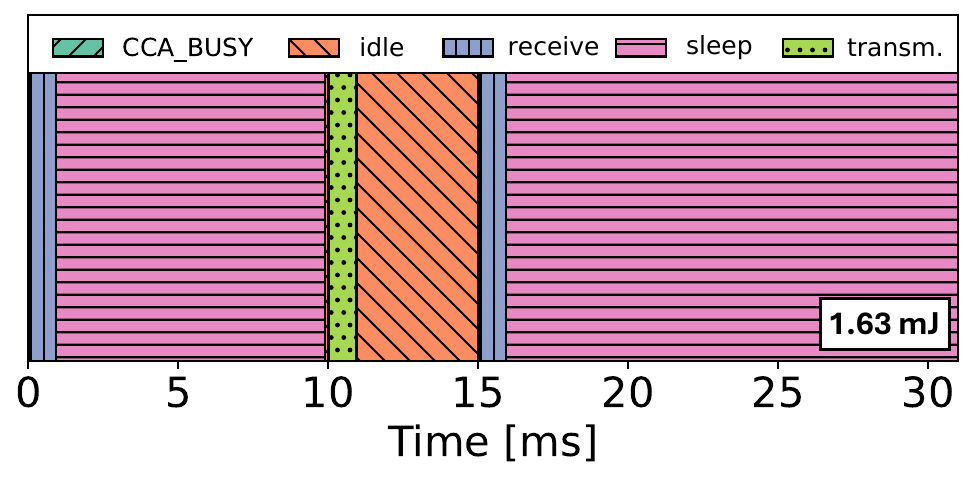}
}
\subfigure[Without TWT]{\includegraphics[width=0.35\textwidth]{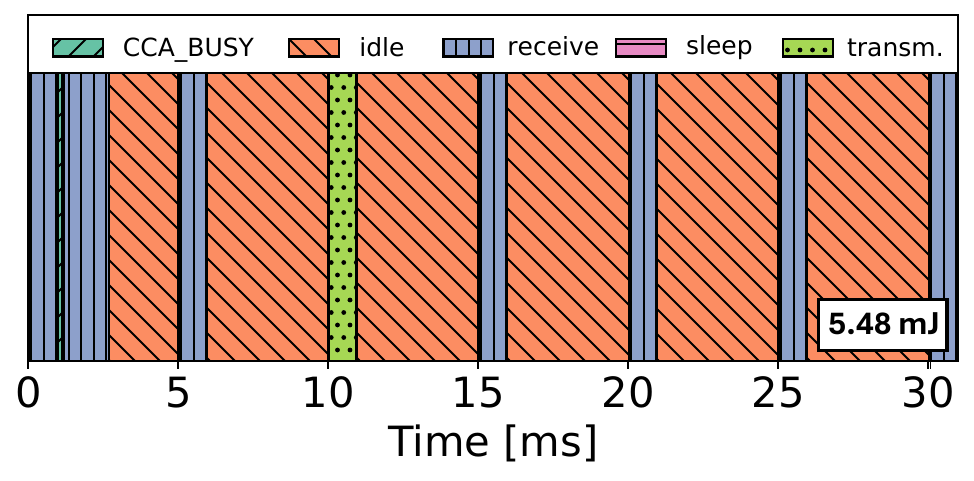}
}
    \caption{Example of the temporal evolution of STA's power states with (a)  and without (b) TWT, over a period of 32\,ms. The total energy consumed during the time period is reported at the bottom right. The example includes 8 STAs, each with a 1,500-byte UDP packet to transmit; the considered STA starts  transmitting at time 10\,ms.}
    \label{fig:stateplot}
\end{figure}

Besides reducing energy consumption, \gls{twt} can enable time-sensitive Wi-Fi networks, i.e., it can provide deterministic latency for the traffic generated by the Wi-Fi nodes. 
Indeed, without TWT, nodes that generate traffic with different requirements and deadlines need to contend for the shared wireless medium. This is especially true when the overall offered traffic is bursty, 
which will make the nodes more likely to compete for channel access, and possibly defer their transmissions.
This makes the timing of channel access and successful data transmissions non-deterministic. Such a level of unpredictability may be unacceptable, especially for time-critical Industry 4.0 use cases  often involving the control of time-sensitive machinery. This concept is demonstrated in Fig.~\ref{fig:twtvsnotwt}, showing the average delay experienced by 8 STAs with TWT enabled and disabled, obtained through our ns-3-twt simulator, as the payload size increases. All STAs generate traffic at the same time. By scheduling non-overlapping windows for each STA, the TWT approach yields a data delivery delay that is almost equal to the case where there would not be any contention, with lower jitter and higher determinism (the confidence intervals are indeed very small). When instead TWT is disabled, the data delay becomes much higher, and its variation very evident. 

In spite of the above advantages, it is important to remark that {\em TWT alone does not provide any guaranteed AoI}, as \mbox{Wi-Fi} packets could wait in the transmission queue for more than what the maximum tolerable AoI allows, before they are scheduled in a TWT \gls{sp}. This is exemplified in Fig.~\ref{fig:naivevsoptimum}, which presents in plot (a) an instance of a problem of TWT acceptance and scheduling at the AP. For a scheduling solution to be feasible, the transmissions (blue bars) cannot overlap in time. Rather, they have to be scheduled sequentially within the white boxes, indicating the transmission ranges that respect the data deadlines. When a naive first-input-first-output (FIFO) approach is used in plot (b), only a few transmissions can be scheduled, while several are not since they would miss their deadline. In contrast, by accounting for the traffic deadlines, the number of transmissions can be maximized (see plot (c)), doubling the number of transmissions that can be accommodated.  
{\color{black} The above observations highlight that, in the presence of energy-constrained devices like IIoTs and strict latency requirements in data packet delivery, {\em it is imperative to develop a scheduling strategy for Wi-Fi networks that accounts for traffic generation times and deadlines}.}

\begin{figure}
  \centering
  \includegraphics[width=0.40\textwidth]{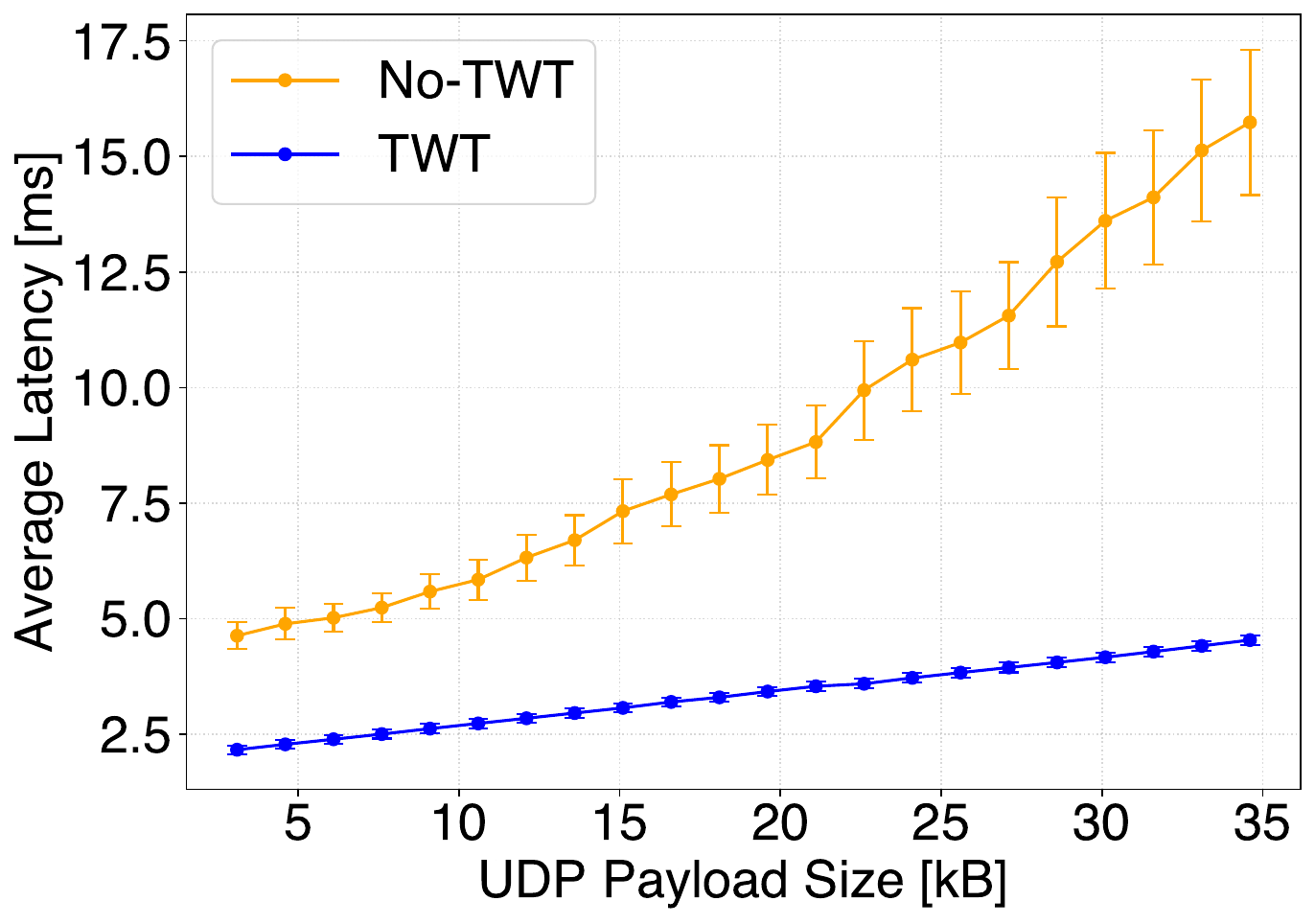}
  \caption{Comparison of the average STA transmission delay measured with and without TWT, when used to avoid channel contention. Error bars show the 95\% confidence intervals.}
  \label{fig:twtvsnotwt}
\end{figure}

\begin{figure}
  \centering
\subfigure[Problem instance]{
\includegraphics[width=0.47\textwidth]{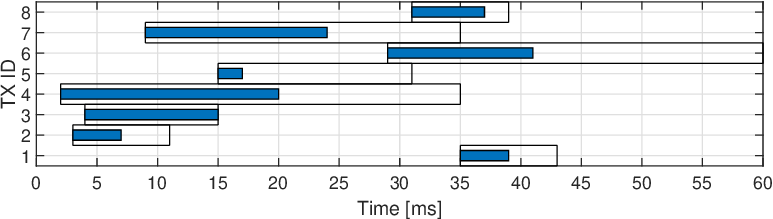}
}
\subfigure[FIFO solution]{
\includegraphics[width=0.47\textwidth]{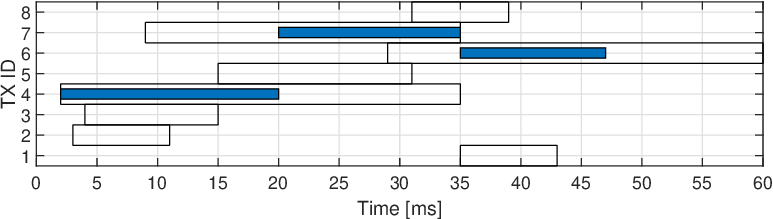}
    }
  \subfigure[Optimum]{\includegraphics[width=0.47\textwidth]{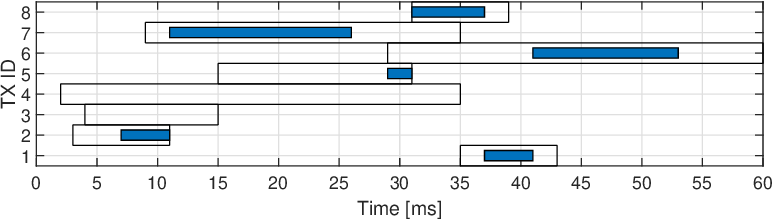}
}
   \caption{TASP example (a) and its solution using a FIFO strategy (b) and a strategy maximizing the number of transmissions (c). The blue bars indicate the transmission times, which have to start and end inside the white boxes to meet the traffic deadlines.}
    \label{fig:naivevsoptimum}
\end{figure}

{\color{black}
\section{System Model and Problem Formulation}
\label{sec:system_model}

This section first introduces our system model and assumptions, and then presents the TASP optimization problem to be solved at the AP.

\subsection{System Model and Assumptions}
\label{sub:model_assumptions}

We consider an  IEEE Wi-Fi\,6  \gls{he} \gls{bss} with  an  \gls{ap} and  $M$  \glspl{sta}\footnote{{\color{black}We consider individual TWT sessions. 
However, the standard \cite{802.11ax} also defines broadcast TWT, where the AP coordinates shared wake times for multiple STAs simultaneously. Moreover, it does not preclude the possibility of establishing TWT agreements directly between STAs, for instance in an IBSS setting.}}, 
with \glspl{sta} having time-sensitive traffic to transmit to the \gls{ap}. To avoid \mbox{non-deterministic} delays caused by channel contention, the \gls{ap} leverages the \gls{twt} mechanism to multiplex \gls{sta} transmissions in the time domain, aiming to schedule as many traffic flows as possible.  
The scheduling period, defined as the time between consecutive scheduling decisions, is bounded by the beacon frame, periodically transmitted every $T_b$ seconds by the \gls{ap}. Accordingly, each \gls{twt} scheduling decision holds valid till the next beacon transmission. A beacon interval is divided into $D$ slots, each of duration $T_s {=} T_b/D$; a slot represents the temporal granularity for \gls{twt} scheduling. This implies that the awake period for a STA during which it can send/receive data, also called Service Period (SP),  begins at the start of a slot and spans over a discrete number of slots. 
Additionally, we assume the propagation delay to be negligible. 

 A \gls{twt} session between the AP and a STA is composed of one \gls{sp} or many periodic \glspl{sp}, if the session is implicit~\cite{802.11ax}. For efficiency, we focus on implicit sessions, as they do not require new scheduling if the traffic pattern does not change. We remark, however, that such sessions can be modified at every periodic interval if the traffic pattern varies. 
Also, our solution applies to any of the  TWT mechanism configurations foreseen by the Wi-Fi standard.


Within a beacon interval, each \gls{sta} $m$ ($m{=}1,{\ldots},M$) may have one or more time-sensitive packets to transfer, possibly belonging to different traffic flows. 
Based on an experienced signal quality level $l_m$, a \gls{sta} $m$ can determine  the maximum data rate $\rho_m(l_m)$ at which it can successfully transmit towards the AP. In addition to signal level, this rate depends both on the \gls{sta} capabilities and on the network configuration (e.g., channel bandwidth, number of MIMO streams).  

 The \gls{sta} then uses the TWT mechanism to request the scheduling of the data transmission for the computed transmission time, i.e., to request a  
 \textit{Transmission Opportunity} (TXOP) in a TWT \gls{sp}. According to the Wi-Fi\,6 
 specifications~\cite{802.11ax}, the \gls{sta} sends  a \textit{Suggest TWT} message to the AP requesting a desired \gls{sp}, and the desired Minimum TWT wake duration needed to accommodate the TXOP. 
Furthermore, 
we assume that the STA  also indicates the  class of the traffic flow to which data belongs, so that the AP is aware of the corresponding traffic priority level. We will  refer to the data, belonging to a given traffic flow that a STA has to transmit, simply as {\em transmission} (TX).
The generic \mbox{$j$-th}  TX is characterized by: \textit{(i)} the source \gls{sta} id $m$, \textit{(ii)} the  amount of data in bytes, $b_j$, \textit{(iii)} the time $g_j$ at which those data were generated at the application layer by the STA, \textit{(iv)} the hard AoI deadline $d_j$ within which the TX has to be received, and \textit{(v)} the traffic priority level $p_j$. 
It follows that, for each TX $j$, a STA $m$ will request a TXOP of  duration equal to $\tau_j{=}b_j/\rho_m(l_m)$, where $b_j$ denotes the amount in bytes of data belonging to a traffic flow $j$ to be transmitted and $\rho_m(l_m)$ the achievable data rate that the  \gls{sta} $m$ can use. 
For simplicity and without loss of generality, we express $g_j$, $d_j$, and $\tau_j$ in time slots.  
Note that all these elements are expressed in relative time units within the beacon interval.

\begin{figure}
  \centering
  \includegraphics[width=\columnwidth]{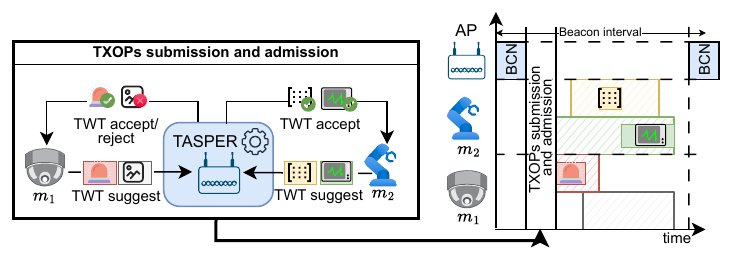}
  \caption{\color{black} System model overview. At each beacon interval, STAs first request traffic scheduling from TASPER, and then follow its scheduling decisions. }
  \label{fig:system_model}
\end{figure}

The AP collects all the TXOP requests at the beginning of the beacon interval and, according to a given strategy, 
it schedules TWT \gls{sp} for the requested TXOP by sending to each requesting STA an  {\em Accept TWT message} {\color{black}(see Fig.~\ref{fig:system_model})}.
Such a message includes some important parameters, namely (i)  the Target Wake Time, i.e., the time instant at which the \gls{sta} should wake up for the TWT session;
(ii) the TWT Wake Interval, i.e., the time interval between subsequent \glspl{sp} for the \gls{sta} (we assume this value to be constant for all \glspl{sta} and equal to the beacon interval);  and (iii) 
the Minimum TWT wake duration, i.e., the minimum duration that the \gls{sta} should stay awake since the  \gls{sp} starts.  

}

Let $\Nc$ be the set of  TXOPs (or, equivalently, TXs) that the AP  is requested to schedule within 
a beacon interval; note that $|\Nc|{\geq} M$. Then the AP will respond to each request in one of the following ways: (i) accept the TWT suggestion as is; (ii) accept the TWT suggestion with a modification; (iii) reject the TWT suggestion. 

\begin{table}[!t]
\centering
\vspace{0.03in}
\caption{Main notation}
\resizebox{0.48\textwidth}{!}{
\begin{tabular}{|c|l|}
\hline
{\bf Symbol} & {\bf Description}
\tabularnewline
\hline
$\Nc$ & Set of schedulable TXs\tabularnewline
\hline
$l_m$ & Quality level of the link between AP and STA $m$ \tabularnewline
\hline
$\rho_m(l_m)$ & Maximum data rate for  STA $m$ \tabularnewline
\hline
$\tau_j$ & No. of time slots needed for TX $j$ \tabularnewline
\hline
D & Number of time slots in a beacon frame \tabularnewline
\hline
$T_b$ & Time interval between consecutive beacon frames \tabularnewline 
\hline
$T_s$ &  Slot duration (time granularity) \tabularnewline
\hline
$g_j$ & Generation time (in slots) for  TX $j$ \tabularnewline
\hline
$d_j$ & AoI deadline (in slots) for TX $j$ \tabularnewline
\hline
$p_j$ & Priority level of TX $j$\tabularnewline
\hline
$E^{\text{tx}}_j$
& Per-slot energy consumption  of the STA performing TX $j$ 
\tabularnewline
\hline
$E^{\text{id}}_j$ & Per-slot energy consumption in idle state of the STA performing TX $j$ 
\tabularnewline
\hline
$E^{\text{st}}_j$& \begin{tabular}{@{}l@{}} Per-slot energy consumption of the STA performing  TX $j$ during mode transitions 
\end{tabular} 
\tabularnewline
\hline
$s_{ij}$ & Auxiliary parameter, equal to 1 if TXs $j$ and $i$ are performed by the same STA; 0 otherwise \tabularnewline
\hline
$x_j$ & TXOP admission binary variable\tabularnewline
\hline
$y_{ij}$ & TXOP sequence binary variable\tabularnewline
\hline
$z_j$ & TXOP end time variable\tabularnewline
\hline
$e_{ij}$ & \begin{tabular}{@{}l@{}}Total energy consumed by the STA performing TX $j$ when $j$ is scheduled immediately after TX $i$ \end{tabular} \tabularnewline
\hline
\end{tabular}
}
\label{tab:notations}
\vspace{-3mm}
\end{table}

As for the per-slot energy consumption of an STA $m$, this is given by
 $E^{\text{tx}}_m {=} P^{\text{tx}}_m {\cdot} T_s $ where  $P^{\text{tx}}_m$ is the transmission power of STA $m$
(assumed to be constant for simplicity).  A TX of  duration $\tau_j$ pertaining to STA $m$  then implies  an energy consumption 
$E^{\text{tx}}_m {\cdot} \tau_j$. Similarly, let $P^{\text{id}}_m$ be the power consumed by \gls{sta} $m$ when idle, 
and   $E^{\text{id}}_m {=} P^{\text{id}}_m {\cdot} T_s$  be the corresponding per-slot energy consumption. 
To save energy, before and after a TX, \gls{sta} $m$ enters the sleep state; we denote with $E^{\text{st}}_m$ 
the energy consumed during the transition from sleep to transmission mode  and vice versa. 
Let $j$ be a TX held by STA $m$. For simplicity, in the following we abuse the notation by replacing  $E_m^{\text{tx}}$, $E_m^{\text{id}}$, and $E_m^{\text{st}}$,
respectively, with $E_j^{\text{tx}}$, $E_j^{\text{id}}$, and $E_j^{\text{st}}$.
Hence, $E_j^{\text{tx}}$, $E_j^{\text{id}}$ and $E_j^{\text{st}}$ represent the per-slot energy consumption incurred, respectively,  during transmission, in idle state, and while transiting between operational states, by the STA performing TX $j$.


{\color{black}\subsection{TASP Formulation}\label{sub:formulation}
Let us now define the TWT Acceptance and Scheduling Problem (TASP), to be solved at the AP. TASP aims to efficiently leverage the Wi-Fi TWT mechanism to meet the stringent demands of IIoT networks—specifically, minimizing energy consumption, ensuring timely data delivery, and handling concurrent traffic flows with diverse priorities and deadlines.}  

Ideally, the AP should schedule all the requested TXOPs, while  keeping  the overall energy consumption as low as possible. 
However,  scheduling all TXs may not be possible, in which case the AP can reject one or more TWT suggestions, e.g., 
those associated with the lowest priority TXs, or the longest ones.  
Choosing  \textit{which} TWT suggestions to reject enables a trade-off between energy consumption and traffic priority level. 
In the following, we formulate the TWT Acceptance and Scheduling Problem (TASP),  
which, depending on the value of a weight  coefficient, $\beta$,  sets the relative importance of two objective components, namely, the cost of rejecting TXs and the STA energy consumption, and aims at minimizing both. 
The decision variables we consider are  as follows:
\begin{itemize}
    \item $\boldsymbol{x} {=} [x_j]$, defined as the \textit{transmission acceptance vector}, where the generic integer element $x_j{\in}\{0,1\}$  indicates whether or not the TXOP associated to TX $j$ is admitted;
    \item $\boldsymbol{Y} {=} [y_{ij}]$,  defined as the \textit{transmission sequence matrix}, whose element $y_{ij}$ takes on 1 when the TXOP associated to TX $j$ is scheduled right after the TX $i$, and 0 otherwise;
    \item $\boldsymbol{z} {=} [z_j]$, defined as the \textit{end time vector} containing the end time of the scheduled TXOPs.
\end{itemize}
Additionally, we introduce 
 two dummy TXs, denoted by $\alpha$ and $\omega$, to allow for a consistent definition of the scheduling sequence of TXOPs, with TX $\alpha$ and $\omega$ being the  first and the last  of the sequence (respectively) and being both associated to null generation time, duration, and traffic priority, and an AoI deadline that is equal to the maximum  among those of all genuine TXs.


We can then write the TASP as:
\begin{usecase}[TWT Acceptance and Scheduling Problem (TASP)]
\begin{mini!}|s|[3]<b>{\boldsymbol{x},\! \boldsymbol{Y},\!\boldsymbol{z}}
{\!\!\sum_{j {\in} \mathcal{\Nc}}\Biggl[ \beta \Bigl((1{-}x_j) p_j \Bigr) {+} (1{-}\beta)  x_j \!\!\!\!\!\!\sum_{i{\in} \mathcal{\Nc}, i \neq j} \!\!\!\! e_{ij} y_{ij}\Biggr]   \label{obj}}{\label{opt}}{}
\addConstraint{\sum_{i,i {\neq} j} y_{ij}}{{=} x_j}{\forall {i} {\in} \{{\Nc} \cup \{\alpha\} \}} \nonumber \addConstraint{}{\label{constr:seq_a}}{\forall j{\in} \{{\Nc} \cup \{\omega\} \}}
\addConstraint{\sum_{i,i {\neq} j} y_{ji}}{{=} x_j}{\forall {i} {\in} \{{\Nc} \cup \{\omega\} \}} \nonumber
\addConstraint{}{\label{constr:seq_b}}{\forall j{\in} \{{\Nc} \cup \{\alpha\} \}}
\addConstraint{z_i {+} \tau_j y_{ij}{+} d_j(y_{ij}{-} 1)}{\leq z_j}{\forall i{\in} \{{\Nc} \cup \{\alpha\} \}}\nonumber
\addConstraint{}{\label{constr:completion}}{\forall j{\in} \{{\Nc} \cup \{\omega\} \}, \, i {\neq} j}
\addConstraint{(g_j+\tau_j)x_j}{{\leq} z_j \label{constr:release}}{{\forall j{\in} \{{\Nc} \cup \{\omega\} \}}}
\addConstraint{z_j}{{\leq} d_j x_j\label{constr:deadline}}{{\forall j{\in} \{{\Nc} \cup \{\alpha\} \cup \{\omega\} \}}}
\addConstraint{e_{ij}}{{\leq} E_j^{\text{st}} {+} \tau_j E_j^{\text{tx}}\label{constr:energy_ub}}{{\forall j{\in} \{{\Nc} \cup \{\alpha\} \cup \{\omega\} \}}}
\addConstraint{e_{ij}}{{\geq} y_{ij}\! \left(\!s_{ij} \cdot \min \bigl\{ E_j^{\text{id}}(\!z_j {-} \tau_j {-} z_i\!),  \right.}{}\nonumber
\addConstraint{}{\left. E_j^{\text{st}}\bigr\} {+} (1 {-} s_{ij}\!) E_j^{\text{st}}\!\right)}{{{\forall i{\in} \{{\Nc} \cup \{\alpha\}\}}}}\nonumber
\addConstraint{}{\label{constr:energy_lb}}{{{\forall j{\in} \{{\Nc} \cup \{\omega\} \}}},\, i {\neq} j}
\addConstraint{z_{\alpha}{=}0, z_{\omega}{=} \max_{j{\in}\Nc}\:\{d_j\}}{{\leq} D\label{constr:dummy_a}}{}
\addConstraint{x_{\alpha} {=} 1, \, x_{\omega} {=} 1}{\label{constr:dummy_b}}{}
\addConstraint{x_j {\in} \{0,1\},\, y_{ij} {\in} \{0,1\}}{}{{{\forall i{\in} \{{\Nc} \cup \{\alpha\}\}}}}\nonumber
\addConstraint{}{\label{constr:binary}}{{{\forall j{\in} \{{\Nc} \cup \{\omega\} \}}},\, i {\neq} j}
\end{mini!}
\end{usecase}

In the above formulation, the cost function (\ref{obj}) is the weighted sum of two elements: (i) the \textit{rejection cost}, defined as the sum of the priority of the TXs associated with the rejected TXOPs, and (ii)  
the energy cost associated with the transmission of the scheduled TXs. 
Specifically, the energy cost associated with  TX $j$  strictly depends on the scheduling order. Let TX $i$ be the TX scheduled immediately before TX $j$, and  $s_{ij}$ be an auxiliary variable that takes on 1 if TX $i$ and TX $j$ are performed by the same STA, and  0 otherwise. Then we define the energy cost associated with TX $j$, $e_{ij}$,  as:
\begin{equation}
\label{eq:e_ij}
e_{ij} {=} \tau_j E_j^{\text{tx}} {+} (1{-}s_{ij})
E_j^{\text{st}} {+} s_{ij} \cdot
\min \left ( E_j^{\text{id}} \cdot ( z_j {-} \tau_j {-} z_i ), E_j^{\text{st}} \right ).
\end{equation}
Indeed, scheduling TX $j$ obviously implies an energy consumption proportional to the duration of TX $j$. Additionally, the STA performing TX $j$ incurs a cost that  depends on the scheduling order: if TX $i$ (i.e., the preceding TX), is performed by a different STA ($s_{ij} {=} 0$), the STA performing TX $j$ will need to transit from sleep to transmit state, thus incurring an energy cost $E^{\text{st}}_j$.
Conversely, if TXs $i$ and $j$ are performed by  the same STA ($s_{ij} {=} 1$), the latter can either  enter the idle state and wait to perform  TX $j$, or  go to sleep and move to transmit state  just before performing TX $j$.
Between these options, the STA will select the one that implies the minimum energy consumption. 
For simplicity and without loss of generality, in the model we neglect the energy consumed by a STA to receive the acknowledgment from the AP (this will instead be accounted for in our performance evaluation).  
%
Finally, note that both the priority and the energy cost of each TX are properly normalized between 0 and 1, in order to be comparable. 

\edits{Constraints (\ref{constr:seq_a}) and (\ref{constr:seq_b}), instead, enforce that the scheduled TXOPs occur sequentially, without overlapping with each other. Constraint (\ref{constr:completion}) ensures that if the TXOP $i$ precedes the TXOP $j$, the latter ends $\tau_j$ slots after the TXOP $i$.  This implies that the end time  of TX $j$ must occur after the end of TX $i$ plus the duration of TX $j$. Crucially, this ensures that no TX can be scheduled before one of its predecessors~\cite{oas}.}
Constraint (\ref{constr:release}) imposes that the end time of a scheduled TXOP does not occur earlier than the sum of the generation time and the duration of the associated TX. \edits{In other words, if a transmission is generated at time $t{=}g_j$~s and requires $\tau_j$~s, the TXOP cannot end before $t{=}g_j {+} \tau_j$~s.} 
The AoI deadline constraint (\ref{constr:deadline}) specifies that the end time of a scheduled TXOP must not exceed the deadline of its associated TX.  
Regarding energy, (\ref{constr:energy_ub}) upper limits  the energy consumption of the STA associated with TX $j$ to the sum of the energy required for transiting from sleep to transmit mode and vice versa, and of the energy consumed while transmitting.  Conversely, (\ref{constr:energy_lb}) provides a lower bound to energy consumption, which depends on whether the considered TXOP is requested by the same \gls{sta} as the preceding TXOP. As already specified before, if the STA is the same, then such a STA remains in an idle state instead of going back to sleep, thereby conserving energy. 
Constraints (\ref{constr:dummy_a}) and (\ref{constr:dummy_b}) define $\boldsymbol{x}$ and $\boldsymbol{z}$ for the dummy TXs. \edits{In particular, the dummy TX $\alpha$ opens the scheduling chain while the dummy TX $\omega$ marks the end of the chain; both have a duration equal to 0\,s and $\omega$ must occur before the end of the beacon interval. It follows that a feasible schedule yields a sequence of TXs that does not exceed the beacon interval.} (\ref{constr:binary}) enforces the elements of $\boldsymbol{x}$ and $\boldsymbol{Y}$ to take a binary value. 
All notations used so far are summarized in Table \ref{tab:notations}.

Importantly, based on~\cite{scheduling_notation}, we can map the TASP into  a single-machine, sequence-independent, completion-dependent batch setup cost scheduling problem with release dates, deadlines, and rejection.  
Indeed, the batch setup cost can represent  the transceiver activation cost for one or  more subsequent TXs, accounting for potential idle periods between TXs (completion-dependent), along with the subsequent deactivation cost.  
Notably, the most  related problems are \textit{(i)} the Order Acceptance and Scheduling (OAS)~\cite{DEWEERDT2021629}, which includes a job setup cost analogous to the activation and deactivation energy in TASP, and \textit{(ii)} the single-machine Job Interval Selection Problem (JISP)~\cite{jisp}, a special case of the OAS that focuses solely on maximizing the number of accepted tasks.  
Both problems can model the AoI constraints for TXs (i.e., maximum completion times for jobs); however, they leave out the energy consumption cost.  
It follows that, to  our knowledge, a formulation equivalent to the TASP has not been previously proposed.  
As far as the TASP complexity is concerned, the following result holds.
\begin{theorem}
The TASP is NP-hard.
\end{theorem} 
\begin{proof}
    The TASP includes both integer and continuous variables, as well as one quadratic constraint (\ref{constr:energy_lb}), and therefore it is a Mixed Integer Quadratic Constrained Programming  (MIQCP) problem, which is known    to be NP-hard~\cite{burer2011milp}. 
    \end{proof}
    It is worth noting that, even when ignoring (\ref{constr:energy_lb}) and the energy term in (\ref{obj}), the simplified problem equivalent to the single-machine Job Interval Selection Problem (JISP) is also  NP-hard~\cite{jisp}.  
While MIQCP solvers such as Gurobi and IBM CPLEX can compute the optimal solution for TASP, 
real-time scheduling decisions must be made on a per-beacon period basis. 
This requirement, combined with the limited computational resources of low-power \glspl{ap}, 
makes solving non-trivial instances of TASP impractical.

\section{The TASPER Algorithm}\label{sec:TASPER_algo}

The NP-hard nature of the TASP problem clearly calls for the design of an efficient heuristic solution that can swiftly solve TASP even in the presence of a large number of transmission requests.   
We thus propose the TASP Efficient Resolver (TASPER), whose design was inspired by the BALAS algorithm,  introduced in~\cite{DEWEERDT2021629}  to solve the OAS problem (see Sec.\,\ref{sec:system_model}).
 The key intuition behind TASPER is to treat the TWT-based TXOP scheduling as a best-path search within a decision graph. Also, TASPER efficiently schedules traffic data transmissions and it differentiates from previous work (including BALAS) as it effectively embeds energy saving in the scheduling problem.

Specifically, TASPER solves the TX scheduling problem at every beacon interval by implementing the following multi-step strategy. The {\em first step}  consists in building a directed decision graph, which, as shown in Fig.\,\ref{fig:twt_graph}, is composed of:

\noindent $\bullet$  A set of vertices, including:  (i) a virtual vertex, $\alpha$, representing the starting point of the sequence of TXs, (ii) a vertex for each TX to be scheduled and whose AoI deadline has not yet expired, and (iii) a virtual vertex $\omega$, representing the end of the TXs sequence. 
For simplicity, given a transmission request TX $j$, we abuse the notation and use it also to denote the corresponding vertex;  

\noindent $\bullet$  A set of edges, each of which connects two possible back-to-back TXs. 
Notice that the edges exiting vertex $\alpha$ and entering vertex $\omega$ are as many as the number of possible TXs to be scheduled. Also, each edge exiting vertex $i$ and entering vertex $j$,  is associated with a bi-dimensional weight defined 
as $\wb_{ij}{=}[\tau_{j},v_{ij}]$ where, as mentioned earlier, $\tau_{j}$  denotes the duration of TX $j$ and $v_{ij}$ is defined according to the TASP objective function, 
i.e., $v_{ij} {=} \beta p_j {+} (1{-}\beta) (1{-}e_{ij})$. (We recall that the priority of TX $j$, $p_j$, 
and the energy cost, $e_{ij}$,  are normalized within the range $[0,1]$).


\begin{figure}
  \centering
\includegraphics[width=0.90\columnwidth]{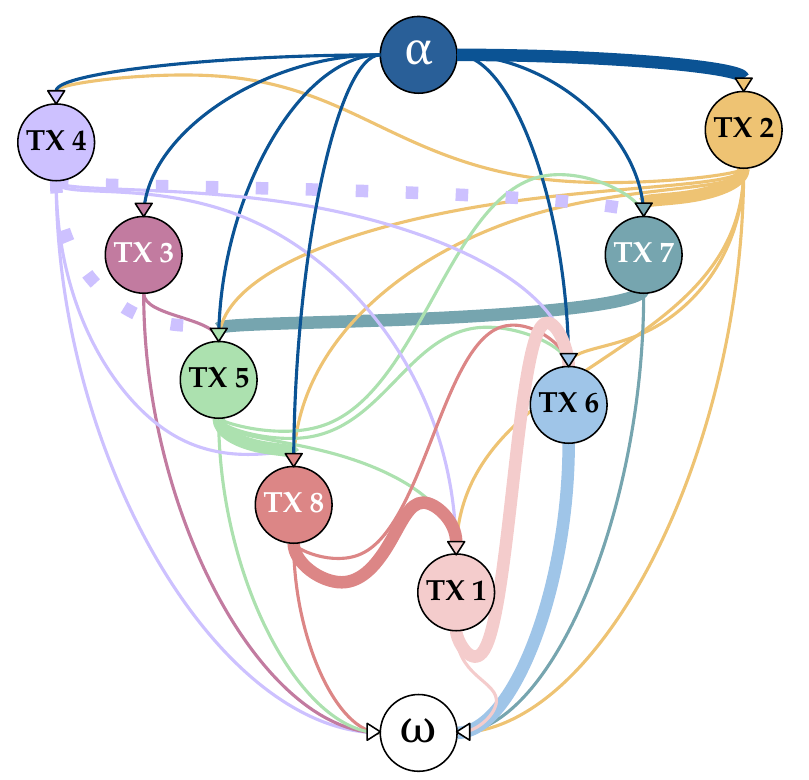}
  \caption{Example of the TASPER  graph, referring to the scheduling problem in Fig.\,\ref{fig:naivevsoptimum}.  
  For clarity, the outbound edges  have the same color as the vertex they exit while their arrow has the same color as the inbound vertex. Notice 
  how the graph is not fully connected:  e.g., there is no edge from TX\,3 to TX\,4, since, if TX\,3 is scheduled, its end time would exceed the latest possible start time of TX\,4.
  The  solid thick lines, traversing vertices $\alpha$, $2$, $7$, $5$, $8$, $1$, $6$, and $\omega$, 
  mark the best path, i.e., the one identified by the Optimum solution 
  in Fig. \ref{fig:naivevsoptimum}. 
  Finally, the dotted thick lines identify the neighborhood of Tx $4$ with a neighborhood size $\eta{=}1$, including TXs 5 and 7.}
  \label{fig:twt_graph}
\end{figure}

Given the above graph,  visiting a vertex means scheduling the TXOP needed to fulfill the corresponding transmission request. 
It follows that scheduling a sequence of TXs can be represented as a path from $\alpha$ to $\omega$.  
Accordingly, as TASPER schedules TXOPs fulfilling the TX requests, i.e., it builds a path traversing the graph, 
the {\em second step} of the algorithm consists in identifying the feasible scheduling solutions among all the possible ones. 
To this end, TASPER keeps track of the residual scheduling time available within the given  beacon interval. 
Let $t_0$ be the variable tracking the elapsed time. 
Let TX $i$ be the last visited vertex, and TX $j$ the potential next vertex to visit. The latter can be visited (i.e., it is a feasible choice) only if the maximum between $t_0$ and $g_j$, plus $\tau_{j}$, does not exceed the AoI deadline $d_j$ of TX $j$: $\max{\bigl\{t_0, g_j \bigr\} + \tau_{j}} \leq d_j$.   
We thus define a {\em feasible path} from $\alpha$ to $\omega$ as a sequence of edges such that, at any vertex $j$ composing 
the path, the intermediate sum of the TXOP  duration  allocated for the upstream vertices  honors the delay constraint of the TX $j$. 
 
 Finally, to enable TASPER to select the best scheduling solution, we define the \textit{value} of a path as the sum 
 of the values of the traversed edges, which is consistent with the expression of TASP's objective function. 
Notice that, given vertex TX $j$, the value taken by the weight of the edge going from TX $i$ to TX $j$ 
depends on whether the STA requesting TX $i$ is the same that is requesting TX $j$. As pointed out 
in Sec.\,\ref{sec:system_model}, the energy consumption $e_{ij}$ of a TX $j$ is indeed influenced by the specific scheduling order adopted.  
TASPER's {\em third step} then consists in  finding the feasible path 
from $\alpha$ to $\omega$  that maximizes the path value defined above.  
To reduce the time complexity of the algorithm, we let TASPER achieve this goal by considering a limited  number of candidate paths. 
To this end,  TASPER exploits the concept of \textit{neighborhood} of a TX and uses the  neighborhood size, 
denoted  by $\eta$, as a tunable parameter to trade off the algorithm optimality gap with its time complexity.  
Let us consider the list of not-scheduled-yet and not-expired TXOP requests, sorted  in ascending order according 
to their latest start time (the latest start time of a TX $j$ is defined as $d_j {-} \tau_j$). 
Then TX $j$ is within the neighborhood of TX $i$ if their respective indexes in the list differ at most by $\eta$. 
An example of a TX's neighborhood is given in Fig.\,\ref{fig:twt_graph}. 
In such a way, indicating with TX $i$  the last scheduled transmission, the search for the next TX to schedule is 
restricted to the closest TX to TX $i$ in the above-mentioned ordered list.

In summary, TASPER's key idea  is that,  at each vertex $i$, it selects as next TX to visit the \textit{dominant} 
TX choice in the neighborhood $\mathcal{N}_i$, i.e., the one that \textit{(i)} is associated with  the highest value and that, 
in case of a tie, \textit{(ii)} yields  the shortest completion time. By doing so, it is possible to select the TXs that have 
both high priority and low completion time, thus leaving more room to schedule additional TXs. 
This provides a solution that maximizes the terms appearing in the TASP objective function as well as the number of scheduled TXs.

\begin{algorithm}[t]
\small
\caption{TASPER algorithm}
\label{alg:TASPER}
\begin{flushleft}
 \hspace*{\algorithmicindent} \textbf{Input} \\  \hspace*{\algorithmicindent} -List of transmissions $L$ (TX $\alpha$ excluded) \\
 \hspace*{\algorithmicindent} -Neighborhood size $\eta$ \\
 \hspace*{\algorithmicindent} \textbf{Output} \\
 \hspace*{\algorithmicindent} -TXOP scheduling $P_{\textbf{best}}$ \\ 
\begin{algorithmic}[1]
\State Order List $L$ by TX latest start time \label{TASPER:order_list}
\State Assign an index $ind_j$ to each transmission in List $L$, equal to the position of TX $j$ in List $L$  \label{TASPER:assign_list_index}
\State Create an empty path List $Pt$ \label{TASPER:initialize_path_list}
\For{TX $j$ in L} \label{TASPER:initialize_paths_start}
    \State  Create an empty Path object $P$  
    \State  Create an empty list $P$.visited\_TX
    \State $P$.visited\_TX.append($\alpha$)    
    \State $P$.visited\_TX.append($j$) 
    \State $P$.cumulative\_reward $\leftarrow$ $P$.cumulative\_reward + $v_{\alpha,j}$
    \State $P$.time $\leftarrow$ $g_j + \tau_j$
    \State FINDPATH(Job List $L$, Path $P'$, Job $j$, Path List $Pt$) \label{TASPER:initialize_paths_end}
\EndFor
\State $P_{\textbf{best}}$ $\rightarrow$ $\max{\{Pt, key = \text{cumulative\_reward}\}}$ \label{TASPER:select_best_path} \\ 
\Return $P_{\textbf{best}}$
\State
\Procedure{FindPath}{Job List $L$, Path $P$, Job $j$, Path List $Pt$} \label{TASPER:procedure_find_path}
   \State $N_j$ $\leftarrow$ FindNeighbors(Job List $L$, $ind_j$, $\eta$) \label{TASPER:find_neighbors}
   \For{Tx $n$ $\in$ $N_j$}    \label{TASPER:check_visited_start}
       \If{$n \in P$.visited\_TX}
       \State \textbf{continue}
       \EndIf \label{TASPER:check_visited_end}
    \State $P'$ $\leftarrow$ $P$ \label{TASPER:create_new_path_start}
    \State $P'$.visited\_TX.append($n$) 
    \State $P'$.cumulative\_reward $\leftarrow$ $P'$.cumulative\_reward + $v_{jn}$
    \State $P'$.time $\leftarrow$ $\max{\{g_n,P'\text{.time}\}}$ + $\tau_n$ \label{TASPER:create_new_path_end}
    \If{isDominated($P'$, $n$.paths)} \label{TASPER:check_dominance_start}
    \State \textbf{continue}
    \Else \label{TASPER:check_dominance_end}
    \For {Path $P''$ $\in$ $n$.paths}  \label{TASPER:remove_dominated_start}
    \If{isDominated($P''$, $P'$)}
    \State $n$.paths.remove($P''$)
    \State $Pt$.remove($P''$)
    \EndIf
    \EndFor \label{TASPER:remove_dominated_end}
    \State $n$.paths.remove($P$) \label{TASPER:replace_P_start}
    \State $n$.paths.append($P'$) 
    \State $Pt$.remove($P$) 
    \State $Pt$.append($P'$) \label{TASPER:replace_P_end}
    \State FINDPATH(Job List $L$, Path $P'$, Job $n$, Path List $Pt$) \label{TASPER:recursive_call}
    \State \textbf{return}
    \EndIf
    
   \EndFor
\EndProcedure
\end{algorithmic}
\end{flushleft}
\end{algorithm}

We show the pseudocode of  TASPER   in Algorithm\,\ref{alg:TASPER}. First, the  requested TXOPs to be scheduled are 
sorted according to their latest start time, $d_j {-} \tau_j$, forming an ordered list of TXs; hence, 
each TX $j$ is assigned an index $ind_j$ corresponding to its position in this list (Lines \ref{TASPER:order_list}--2). 
Then TASPER creates an empty Paths list to store the various paths from $\alpha$ to $\omega$  and initializes a 
Path for each available TX $j$ (Line \ref{TASPER:initialize_path_list} and Lines \ref{TASPER:initialize_paths_start}--\ref{TASPER:initialize_paths_end}). Notice that each path is initialized with a cumulative value equal to the value associated with TX $j$, calculated as discussed earlier in this section. Finally, the path time is set to $g_j + \tau_j$, i.e., the generation time of the TX plus the TX duration. For each path, TASPER then applies the FINDPATH procedure, as defined in Line \ref{TASPER:procedure_find_path}. This procedure plays a key role. First (Line \ref{TASPER:find_neighbors}), it finds the neighbors of TX $j$ through another auxiliary procedure, FINDNEIGHBORS, which simply applies the neighborhood mechanism previously described. As remarked above, both already scheduled TXs and any TX whose AoI deadline has expired are obviously excluded from the neighborhood. Then, for each TX $n$ identified as a neighbor of TX $j$, the procedure checks whether that TX has already been visited in Path $P$: if so, it moves to the next neighbor TX (Lines \ref{TASPER:check_visited_start}--\ref{TASPER:check_visited_end}); otherwise, the procedure includes TX $n$ in Path $P$  and calculates the new path $P^{'}$'s value  and time component (Lines \ref{TASPER:create_new_path_start}--\ref{TASPER:create_new_path_end}). Subsequently, the procedure verifies whether $P^{'}$ is dominated by another, already discovered path that includes the same TX $n$ (Lines \ref{TASPER:check_dominance_start}--\ref{TASPER:check_dominance_end}). If so, $P^{'}$ is discarded and the procedure jumps to the next neighbor TX; otherwise, any path dominated by $P^{'}$ is removed from the Path list (Lines \ref{TASPER:remove_dominated_start}--\ref{TASPER:remove_dominated_end}). In such a case, the old path $P$ is replaced by $P^{'}$ in both the lists  including TX $n$, and the overall Path list (Lines \ref{TASPER:replace_P_start}--\ref{TASPER:replace_P_end}). Finally, the procedure is called recursively until the path reaches $\omega$ (Line \ref{TASPER:recursive_call}).
After the procedure has identified the feasible paths, the path with the largest cumulative value is selected 
(Line \ref{TASPER:select_best_path}). 

{\bf TASPER Complexity.} 
The time complexity of TASPER can be computed following a similar procedure to~\cite{DEWEERDT2021629}. 
It is given by $O(n \cdot \eta^3 \cdot \sigma \cdot 4^{\eta})$, where  $n$ is the total number of TXs to be scheduled within a given beacon interval, while $\sigma$ 
is  the so-called \textit{slack}, i.e., the maximum number of TXs ready for scheduling across the different time instants in the beacon interval. 
As an example, in Fig.\,\ref{fig:naivevsoptimum}, we have $\sigma{=}3$: indeed,  after about 4\,ms, and, again, at time instant 15\,ms, 3 TXs are available for scheduling (namely, TX 2, TX 3, and TX 4 at 4\,ms, and   TX 4, TX 5, and TX 7 at 15\,ms). 
\edits{Also, to better emphasize the scalability of TASPER with the number of STAs $N$,  consider a given beacon interval and traffic pattern such that each STA generates $P$ packets per beacon interval. In this case, the total number of transmissions to be scheduled  is given by  the number of STAs $N$ multiplied by the number of packets $P$.  By fixing the algorithm parameter $\eta$, the resulting complexity becomes $O(N \cdot P \cdot \sigma)$, i.e., 
 it increases linearly with the number of STAs.
For instance, as the number of STAs grows from 32 to 128,  the algorithm complexity  increases by a factor of four.}

\section{Simulation Framework and Evaluation Methodology}
\label{sec:methodology}

This section first describes ns-3-twt -- the TWT ns-3-based simulation framework we developed and used for our 
performance evaluation. Then it introduces our evaluation methodology and simulation settings, 
as well as the benchmark schemes against which we compare the performance of TASPER.  

\subsection{The ns-3-twt simulation framework} 
\label{sub:sim_ns-3-twt} 
Despite the great interest received by TSN, there exist just a few open frameworks that allow for a reliable simulation of 
such networks, including a realistic model of IEEE 802.11ax with the TWT mechanism. Among these, to our best knowledge, the most complete existing open-source 
solution 
was the one developed by Venkateswaran et al.~\cite{ns-3-twt-shyam-2024-accepted-paper} based on the well-established ns-3 framework~\cite{henderson2008ns-3, ms-van3t-journal}, while other network simulators such as MATLAB or OMNeT++ have not been extended yet to include TWT functionalities. 
We therefore adopted the simulator in \cite{ns-3-twt-shyam-2024-accepted-paper} and enhanced their solution by realizing an advanced simulation framework for TWT and TSN scenarios.

Our simulator, ns-3-twt, is released under an open-source license\footnote{The link will be made available upon acceptance of the paper.}. Building on the already existing implementation of a subset of TWT agreement messages (specifically, the implicit and individual agreements), ns-3-twt introduces a range of enhanced features, including:
{\em (i)} a flexible, complete TWT configuration (e.g., traffic arrival time, TWT Service Period start and duration), {\em (ii)} 
the possibility to set AoI deadlines for the traffic and to log whether they are met or not, {\em (iii)} 
a customized IEEE 802.11ax simulator to easily set up simulation with IEEE 802.11ax and TWT, {\em (iv)} the 
support for multiple classes of STAs, each with a different energy model, {\em (v)} 
the seamless logging of the value and distribution of several performance 
metrics, {\em (vi)} the possibility to set up multiple TWT networks, connecting  the different APs through configurable wired links.

We also highlight that a relevant feature of our simulator is the computation of the energy consumed by the STAs, leveraging an 
accurate energy model that differentiates the energy consumed in the different PHY-layer states (i.e., transmission, reception, 
Clear Channel Assessment, sleep, idle). The energy model in~\cite{ns-3-twt-shyam-2024-accepted-paper} has been extended to support different classes of STAs in the same simulation. Each class has its own energy figure for each of the PHY layer states mentioned above.
The simulator can report the total energy consumed in each of the PHY layer states for each STA, and several other aggregated energy metrics. 

\begin{figure*}[!ht]
  \centering
\subfigure[ns-3-twt TWT scheduling parameters]{
\includegraphics[width=0.42\textwidth]{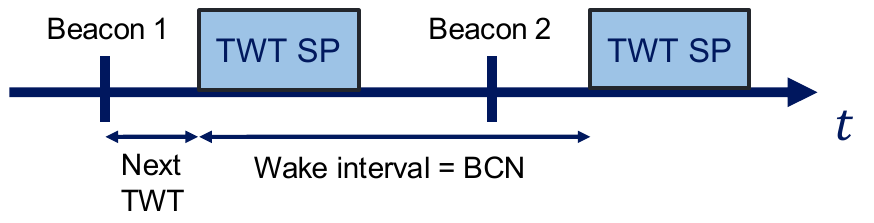}
}
\subfigure[Illustrative example of a simple TWT application with 3 STAs]{\includegraphics[width=0.5\textwidth]{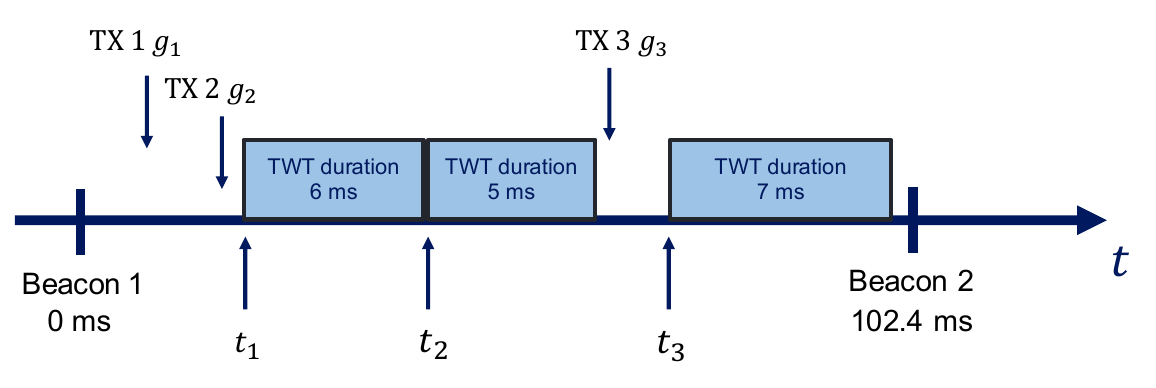}
}
   \caption{ns-3-twt TWT implementation with the main scheduling parameters and an illustrative example of an application with three STAs that have three different transmissions (TX) with different generation times ($g_j$), TWT SP durations and TWT SP start times ($t_j$). The latter are defined through the Next TWT parameter shown on the left. Notice how the TWT scheduling is periodic unless specified otherwise, and the periodicity, also named \emph{wake interval}, corresponds to the beacon interval (BCN).}
     \label{fig:twt}
\end{figure*}

In ns-3-twt, the AP unilaterally dispatches a TWT Accept frame to all STAs based on a specified TWT schedule, to schedule their awake and sleep times.
Within these frames, a crucial part is the \textit{next TWT} field, indicating the time after the start of the next beacon 
interval when the STA can exit the sleep mode and access the physical channel within its TWT \gls{sp}. The window lasts for a 
specific duration denoted by the \textit{TWT duration}. Importantly, in our implementation, the TWT \glspl{sp} are periodic. 
As exemplified in Fig.~\ref{fig:twt}(a), this means that, unless a new schedule is provided by the AP, the STAs will periodically use the same TWT \gls{sp} starting time in each beacon interval. 
Fig.~\ref{fig:twt}(b) gives an illustrative example of a potential TWT application on ns-3-twt. 
The sample scenario includes an AP and three STAs, each of which has a burst of traffic to transmit; the figure  highlights the traffic 
generation times ($g_{j}$), the TWT \gls{sp} start times ($t_{j} {=} z_{j}-\tau_{j}$), and the respective durations for 
each STA within a single beacon interval.

\revision{Note that, by default, ns-3-twt considers a scenario with no other interfering networks. Thus, when scheduling only one STA at a time for each SP, no packet losses or retransmissions are expected to occur, unless the STA and the AP are  too far away from each other, or TWT is not enabled. Instead,  when multiple STAs are scheduled in the same SP, they may   contend for the channel and their transmissions may collide. Investigating this  scenario, however, is out of scope of our analysis, as we focus on  individual TWT sessions. 
It is also important to distinguish between  losses due to channel propagation conditions or channel contention, and  those due to  traffic that could not be scheduled by its target deadline. We remark that the latter type of losses is  taken into account in the simulations.}

\subsection{Evaluation methodology}\label{sub:method}

\begin{figure*}[hb!]
  \centering
\subfigure[]{
\includegraphics[width=0.3\textwidth]{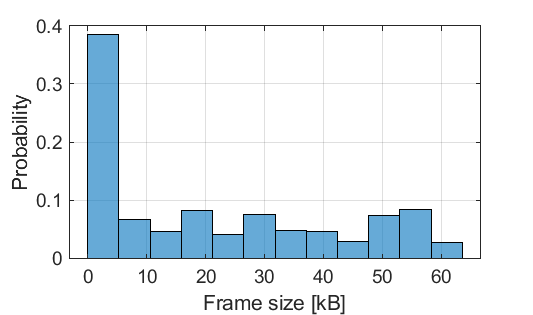}
}\subfigure[]{
\includegraphics[width=0.3\textwidth]{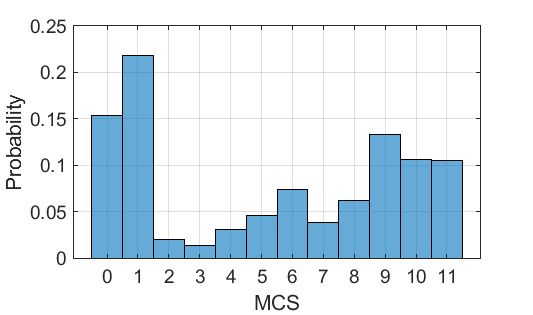}
}\subfigure[]{
\includegraphics[width=0.3\textwidth]{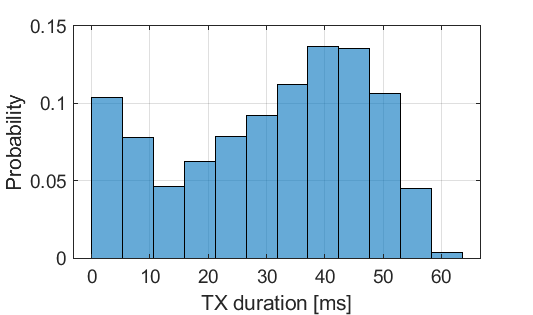}

    }
   \caption{Distribution of: the frame size (a), the best value of MCS that can be used for TXs (b), and the TXOP duration (c), 
   for the generated TASP instances.}
      \label{fig:dist_mcssize}
\end{figure*}

To assess the performance of TASPER, we follow three main steps: {\em (i)}  we first focus on generating some non-trivial 
instances of the TASP; {\em (ii)}  we then use TASPER, and the selected benchmarks, to derive the solutions to the instances generated in the first step; {\em (iii)}  we evaluate the actual performance of such solutions using the ns-3 network simulator. 

To generate non-trivial instances of the TASP, we started by considering three possible scenarios, with 16, 32, and 64 \glspl{sta}, 
respectively. 
In all experiments, each \gls{sta} $j$ requests a TXOP associated with a single TX $j$.  
The TX parameters follow the criteria applied in~\cite{DEWEERDT2021629}, with parameters $\tau {=} 0.2$, and $R{=} 0.3$.
However, since the approach in~\cite{DEWEERDT2021629} accounts only for the transmission durations, in our case 
it is necessary to map the TXOP durations into corresponding values of \gls{mcs} and TX frame size. To  this end, 
the frame size is selected by sampling an empirical distribution shown in Fig.~\ref{fig:dist_mcssize}(a). 
Given the frame size, the MCS is determined as the one providing the frame transmission duration that best approximates 
the requested TXOP. For the generated problem instances, Fig.~\ref{fig:dist_mcssize}(b) and Fig.~\ref{fig:dist_mcssize}(c), 
respectively, show the resulting distributions of STA MCSs and of the TXOP durations. 
Then, we fixed the number of STAs  and evaluated all schemes under study over 100 consecutive beacon intervals. 
By doing so, the TXs not scheduled during a beacon interval can be scheduled in the following beacon intervals, 
provided that their AoI deadline has not expired. 
As for running TASPER on the generated problem instances, we set ${\eta} {=} {9}$, as we found this value to offer a 
good trade-off between solving time and optimality gap. 

\revision{Finally,  as mentioned, in our simulations there are no other ``interfering'' networks. The real-world experimental evaluation (Section~\ref{sec:exp}) instead considers a real environment including other interfering networks operating at 2.4\,GHz.}

\subsection{Benchmarks}\label{sub:benchmar}
We compare TASPER  against the following  strategies:

    \textbullet\, {\em Optimum:} it is the solution calculated by Gurobi, a well-known numerical optimization solver;
    
    \textbullet\, {\em ShortestFirst (SF):} it prioritizes the shortest TXOP requests. At first, SF initializes the auxiliary variable $t_0 {=} 0$ and considers the subset of all requested TXOPs such that
    duration $\tau_j$ meets the condition: $t_0{+} \tau_j  {\leq}  d_j$ and $t_0 {+} \tau_j {\leq} T_b$. This implies that the TXOP end time exceeds neither the TX AoI deadline $d_j$ nor the end of the beacon interval $T_b$. Among these TXOPs, it schedules the one with the shortest requested duration.  In case of ties, it selects one TXOP according to the following decision criteria (in descending  order of importance): (1) earliest TX AoI deadline, (2) highest TX traffic priority, and (3) oldest TX generation time. If  any tie still exists, it selects a TXOP randomly. 
 
  \textbullet\, {\em FIFO:} it schedules the requested TXOPs according to their  generation time. 
    If, upon serving a TXOP request, the associated AoI deadline has expired, the TXOP request is dropped.  
   Ties are solved based upon  (1) the shortest  TXOP duration, (2) the highest traffic priority; otherwise, the TXOP request is selected at random. 
   
    \textbullet\, {\em PriorityFirst (PF):} it works similarly to SF, but it privileges the TXOPs associated  with the highest traffic priority. In case of ties, it schedules TXOP requests (in descending order of importance): (1) nearest TX AoI deadline, (2) shortest TX duration, (3) oldest TX generation time,  else (4) at random; 
    
     \textbullet\, {\em Random:} it randomly schedules the requested TXOPs with  duration $\tau_j$  such that $t_0 {+}\tau_j   {\leq}  d_j; t_0  {+} \tau_j {\leq} T_b$.  

\revision{ \textbullet\, {\em WirelessHART-based (HSA):}  
an adaptation of the DC-HSA algorithm from \cite{chen2018joint}, originally designed for WirelessHART \cite{iec62591_2016_wirelesshart} industrial 802.15.4 networks.  The original DC-HSA algorithm aims to minimize the sum of end-to-end delivery times of the scheduled traffic flows weighted by their priority level, while ensuring that each flow meets its deadline. More specifically, DC-HSA works by iteratively solving Maximum Weighted Independent Set (MWIS) scheduling problems. For each scheduled flow, DC-HSA assigns:  i) a multi-hop link towards the network sink; ii) a frequency channel; iii) a transmission slot. Our implementation (HSA) preserves the core greedy MWIS scheduling logic, with three key adaptations to comply with the Wi-Fi uplink model: 
(i)  the network topology is a single‐hop star;
(ii) the entire frequency channel is allocated to the STA to which  access is granted; (iii) allocated time intervals are 1.024\,ms long (102.4 ms beacon interval ÷ 100) instead of  10\,ms as in WirelessHART.}

It is also worth mentioning that, to get statistically meaningful performance results,   we ran the Random strategy 100 times on each of the 100 instances. 
We then apply the solution to the problem instances that we obtain under the different solution schemes 
in our ns-3-twt simulator  and derive an extensive set of results, as 
shown below.

\section{Numerical Results}\label{sec:peva}
We consider an industrial sensor network scenario, where 16 to 64 STAs (STAs) are connected to an AP on a 2.4\,GHz 
frequency band, with SISO 20~MHz channel setup for both the STAs and the AP. The setup also employs OFDMA with a 800-ns 
guard interval, and customizable uplink traffic.
In our scenario, the uplink traffic corresponds to packets of different lengths, depending on the problem instance. 
More specifically, the traffic that needs to be transmitted by each STA ranges from 1,600\,B to 64,500\,B, and it 
is fragmented into packets of 2,304\,B when the total traffic that needs to be transmitted exceeds such a length 
(corresponding to the IEEE 802.11 Maximum Transmission Unit).

We then focus on one beacon interval of 102.4\,ms, and consider that each STA may belong to one of four possible classes 
of STAs, each class characterized by different energy parameters, as presented in Table~\ref{tab:energySTA}.  

As mentioned, we tested TASPER and the benchmark strategies on 100 different problem instances (with Random being run 100 times on each instance to guarantee the statistical validity of the experiments). The final results have then been obtained by averaging the output of each problem instance (in terms of mean rejection cost, total energy, and other metrics), and by computing the 95\% confidence intervals, shown in the figures as black vertical bars.

\begin{figure*}[!hb]
  \centering
\subfigure[16 STAs]{
\includegraphics[width=0.31\textwidth]{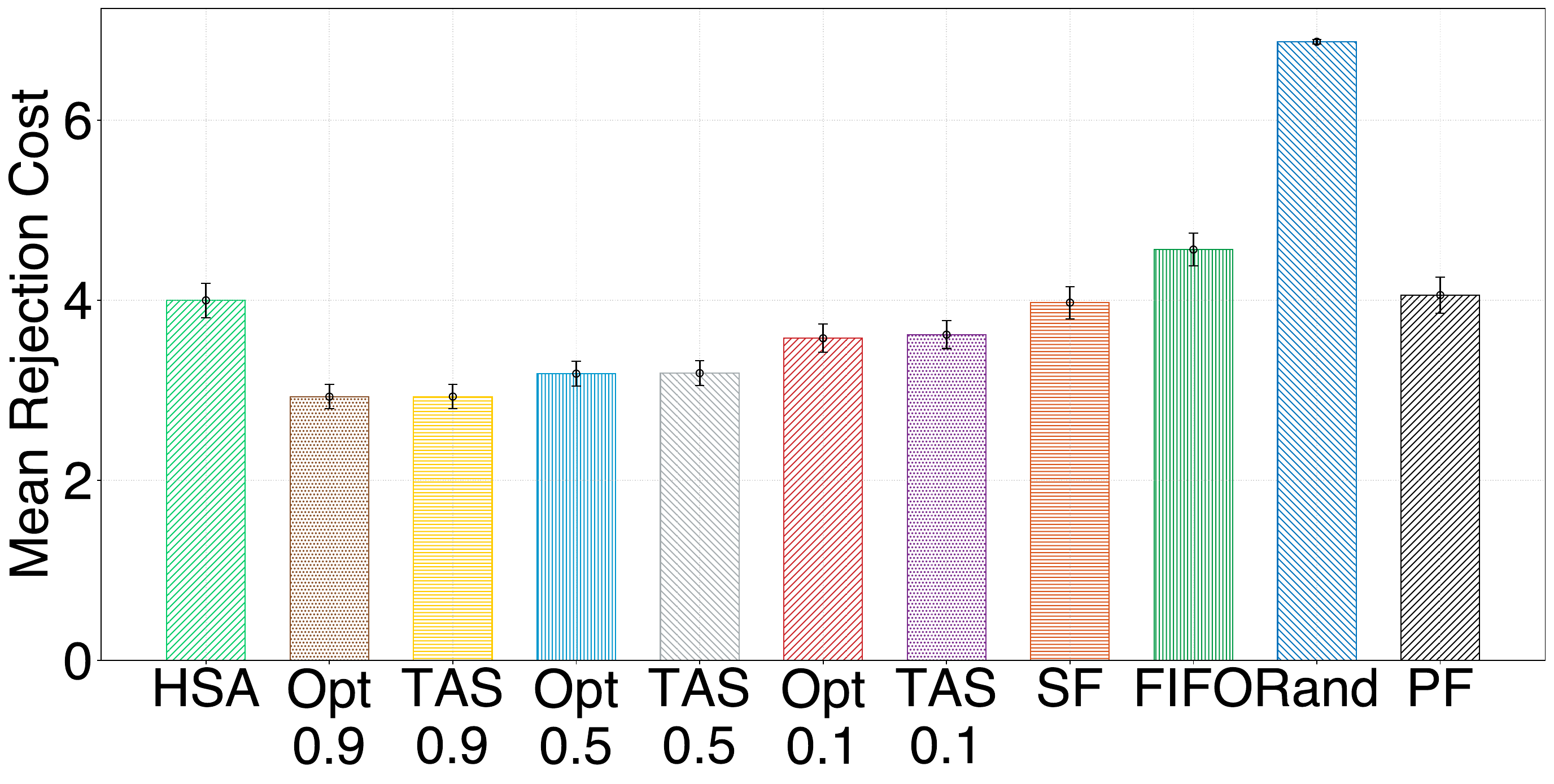}}
\subfigure[32 STAs]{
\includegraphics[width=0.3\textwidth]{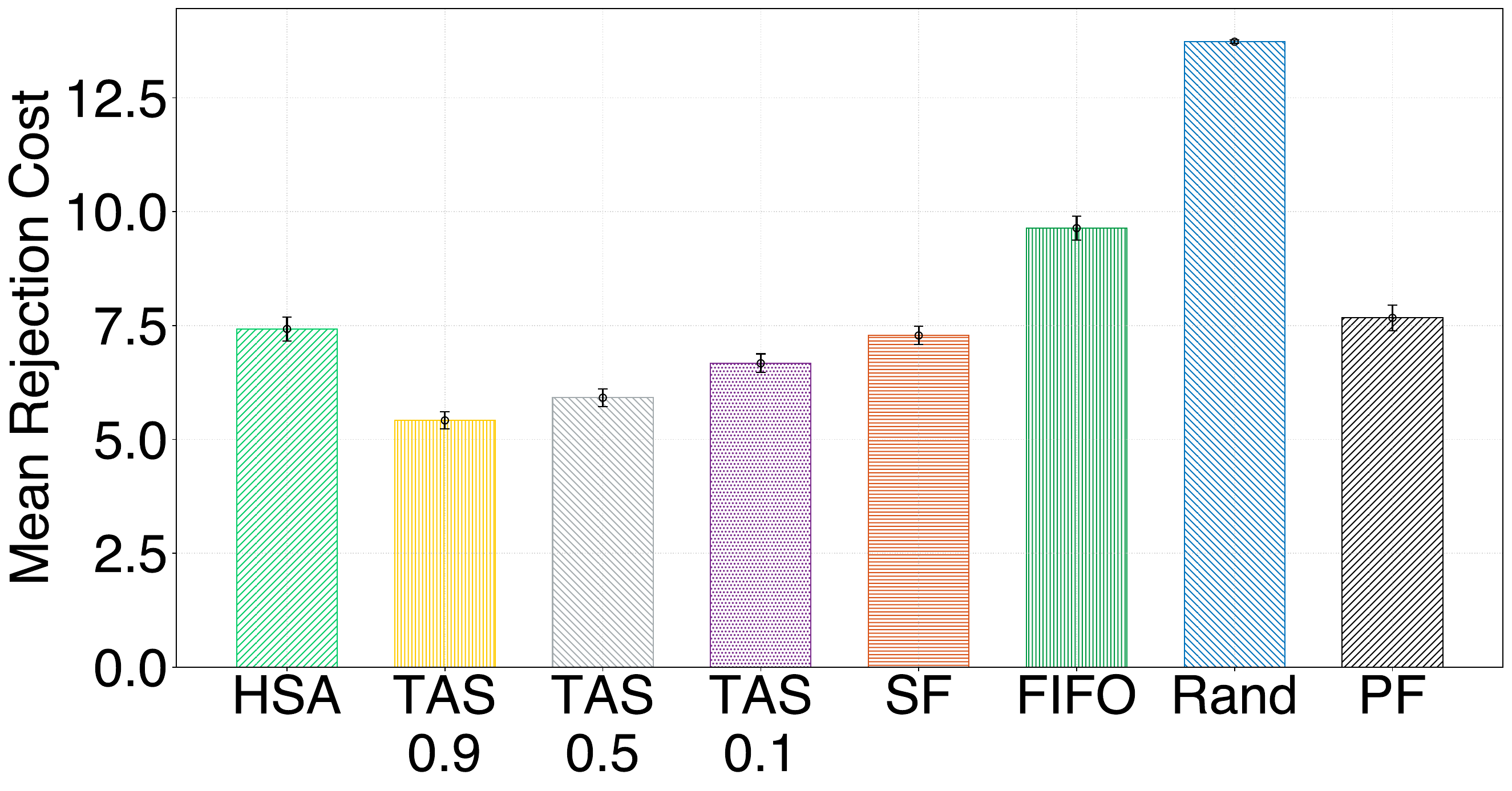}}
\subfigure[64 STAs]{
\includegraphics[width=0.29\textwidth]{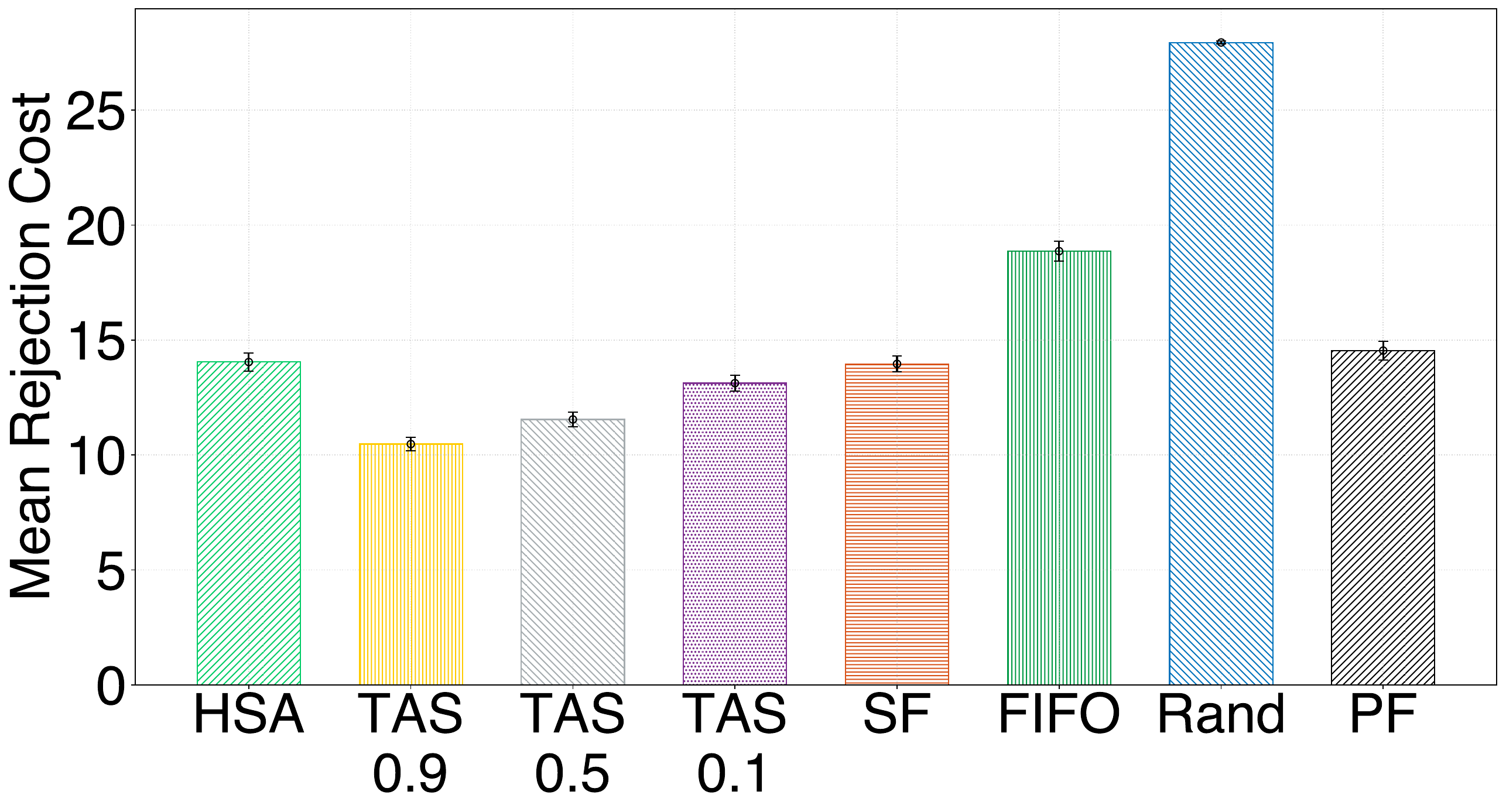}}
   \caption{Mean rejection costs, with 95\% percentiles, for TASPER and its benchmarks, considering progressively 
   complex scenarios, from (a) 16 STAs, to (b) 32 and (c) 64 STAs. From left to right, we compare TASPER (TAS) 
   to HSA, Optimum (Opt), ShortestFirst (SF), FIFO, Random (Rand), and PriorityFirst (PF) strategies. 
   The numbers under ``Opt'' and ``TAS'' represent the value to which we set $\beta$.}
     \label{fig:rejectioncost}
\end{figure*}

We now consider as performance metrics the rejection cost and the mean energy consumption, where we recall that 
the former is defined 
as the cumulative sum of the traffic priority of rejected (therefore, not scheduled) TXs while the latter is computed over 
all STAs in the network. We remark that the unscheduled STAs will remain in sleep for the entire considered period\footnote{\revision{When a STA does not receive a SP  assignment, it remains in sleep mode until the next scheduling period (in our formulation, the beacon interval). To improve efficiency, STAs do not need to resubmit their scheduling requests at the beginning of each scheduling period; instead, the AP considers the requests as valid until their deadline expires.}}.
Figures~\ref{fig:rejectioncost} and~\ref{fig:energyconsumption} present the behavior of such metrics under 
the different strategies, as the number of STAs varies.
It is worth remarking that, due to the very high complexity level of the solution, the results for the Optimum strategy 
could be computed only for 16 STAs, and, even in this case, the time needed to obtain a solution through a popular solver 
like Gurobi greatly exceeded the beacon interval duration. This further highlights that computing the Optimum is 
unfeasible in a practical scenario, fostering the need to devise fast and efficient heuristics like TASPER.

\begin{table}[!t]
\centering
\caption{Typical peak current draw in mA in different operational modes,  associated with the considered STA classes}
\begin{tabular}{|c|c|c|c|c|c|}
\hline
\textbf{STA class} & \textbf{idle} & \textbf{CCA\_BUSY} & \textbf{receive} & \textbf{transm.} & \textbf{sleep} \\ \hline
1  & 50   & 50   & 66    & 232   & 0.12  \\ \hline
2  & 40   & 40   & 40    & 140   & 0.004 \\ \hline
3   & 358  & 358  & 472   & 573   & 12    \\ \hline
4   & 294  & 294  & 388.4 & 555.29 & 11.63 \\ \hline
\end{tabular}
\label{tab:energySTA}
\end{table}

Looking at  the mean rejection cost in Fig.~\ref{fig:rejectioncost}(a), one can observe that TASPER provides 
solutions that are very close to the optimum, with a mean rejection cost just 0.04\% higher than Opt for $\beta{=}0.9$. 
Changing the value of $\beta$ for TASPER,  it is possible to obtain a mean rejection cost that varies from 3.6 
for $\beta{=}0.1$, to 3.2 for $\beta{=}0.5$, reaching 2.9 for $\beta{=}0.9$. The latter corresponds to a reduction of 
19.4\% with respect to $\beta{=}0.1$.
Although ShortestFirst \revision{and HSA also give good performance, they still fall behind TASPER by 9\% and 10\%, respectively}, even when considering the 
worst performance yielded by TASPER (i.e., for $\beta{=}0.1$). 
Similarly, TASPER outperforms HSA by achieving up to 10\% reduction in the mean rejection cost. \revision{In fact, unlike HSA, TASPER can  strategically delay high-priority flows with looser deadlines to accommodate lower-priority ones with tighter deadlines, provided all deadlines are still met. This feature allows TASPER to minimize rejected transmission requests   more effectively than HSA, which schedules transmissions greedily without exploring such trade-offs. 
}

Among all benchmark strategies,  Random is the one 
that performs worst, as it does not employ any logic to maximize the number and priority of accepted TXOPs, 
nor does it take into account energy consumption. Comparing the case with 16 STAs to the more complex scenarios with 32 and 64 STAs 
(Fig.~\ref{fig:rejectioncost}(b) and Fig.~\ref{fig:rejectioncost}(c), resp.), one can notice how the overall mean 
rejection cost increases, since the higher the number of TXOPs to be scheduled, the higher the probability that some of 
them cannot be accommodated.
However, the trend exhibited by the mean rejection cost remains consistent across all scenarios, and TASPER outperforms 
its benchmarks with a gain of 25\% over ShortestFirst, \revision{26\% over HSA}, 28\% over PriorityFirst, 44\% over FIFO and 63\% over Random, 
for $\beta{=}0.9$ and 64 STAs. When considering $\beta{=}0.5$ to improve energy consumption, and keeping 64 STAs, 
TASPER still outperforms its baselines, reducing the mean rejection cost by 18\% with respect to ShortestFirst and 59\% with 
respect to Random.

\begin{figure*}[!ht]
  \centering
\subfigure[16 STAs]{
\includegraphics[width=0.31\textwidth]{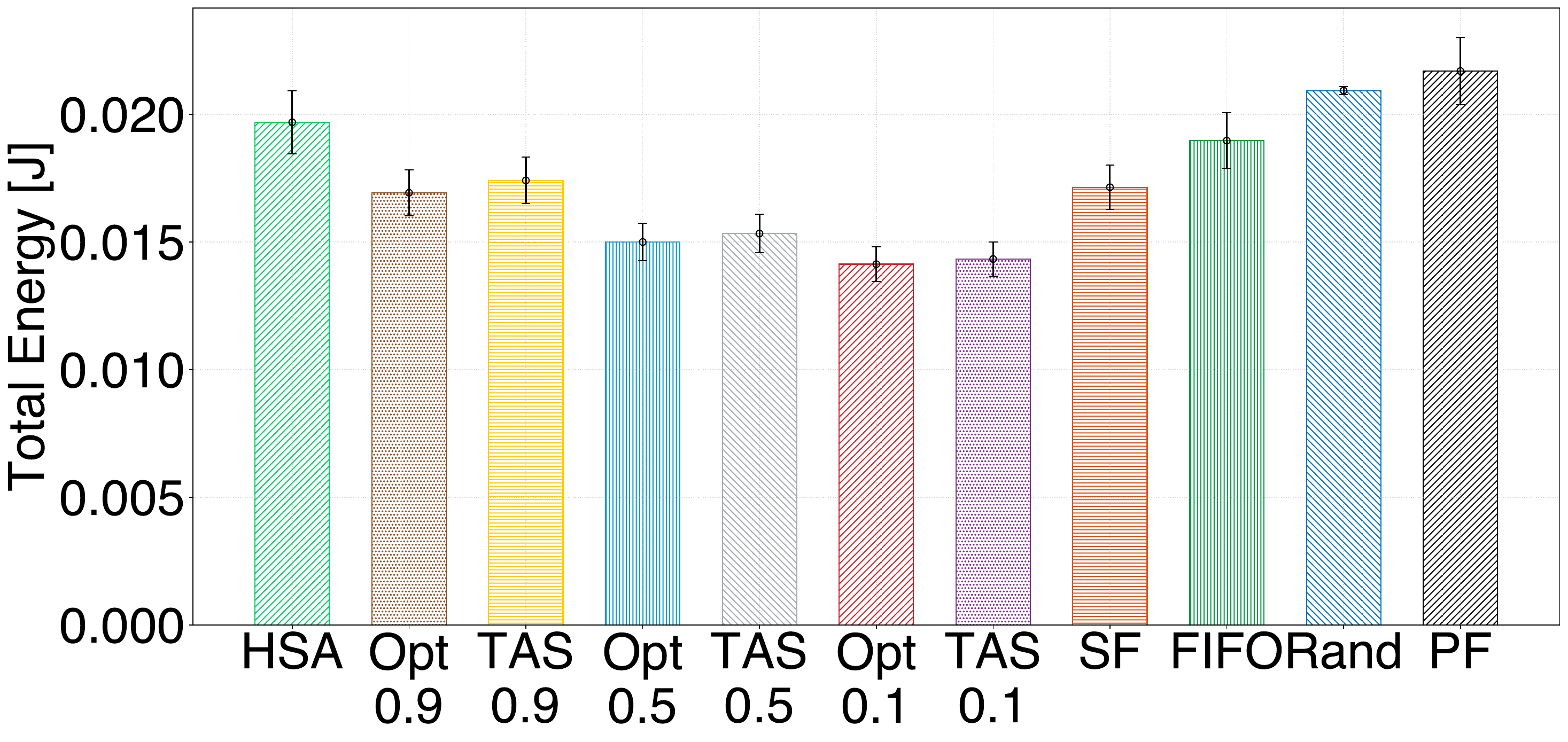}}
\subfigure[32 STAs]{
\includegraphics[width=0.29\textwidth]{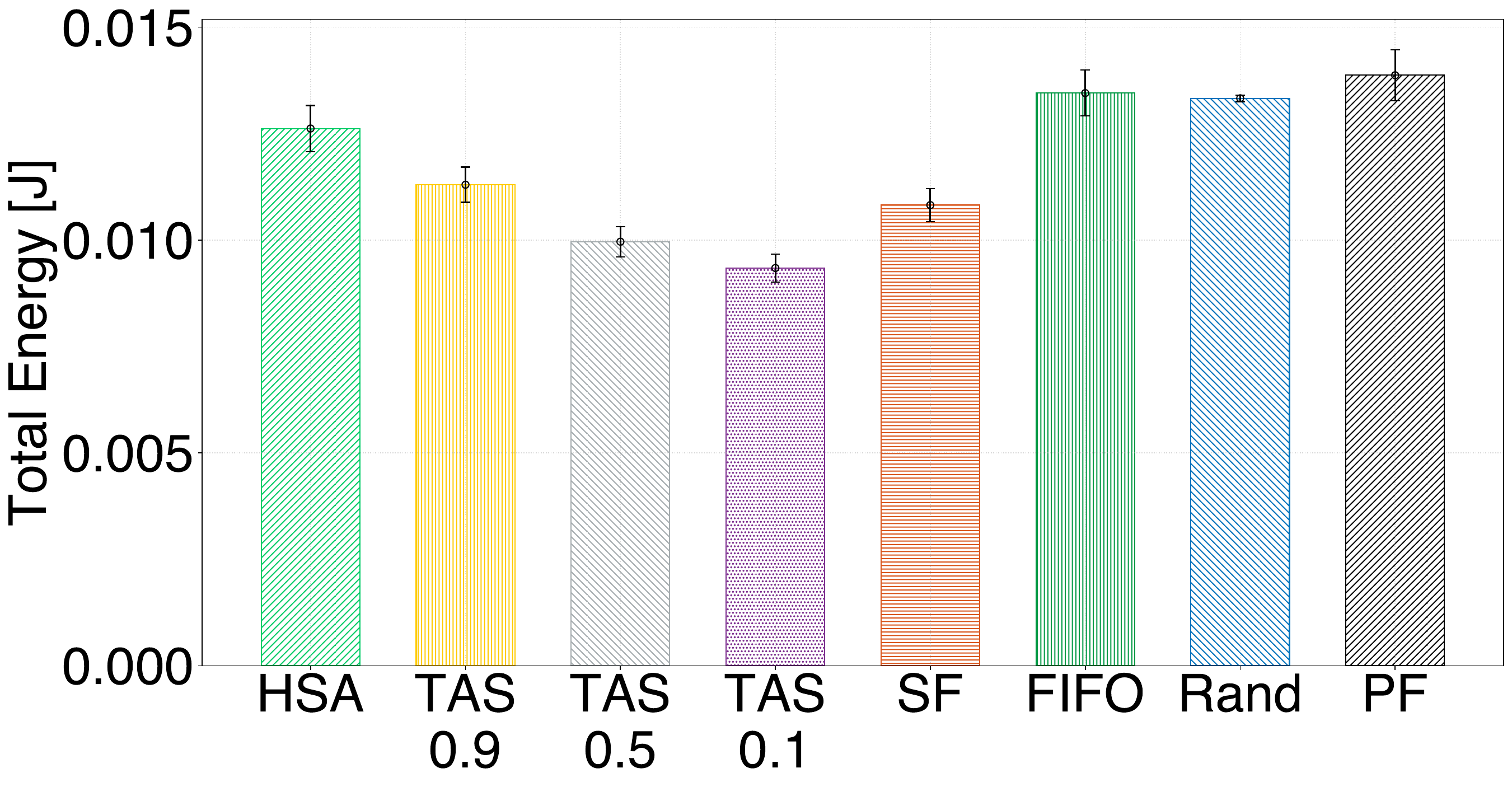}}
\subfigure[64 STAs]{
\includegraphics[width=0.29\textwidth]{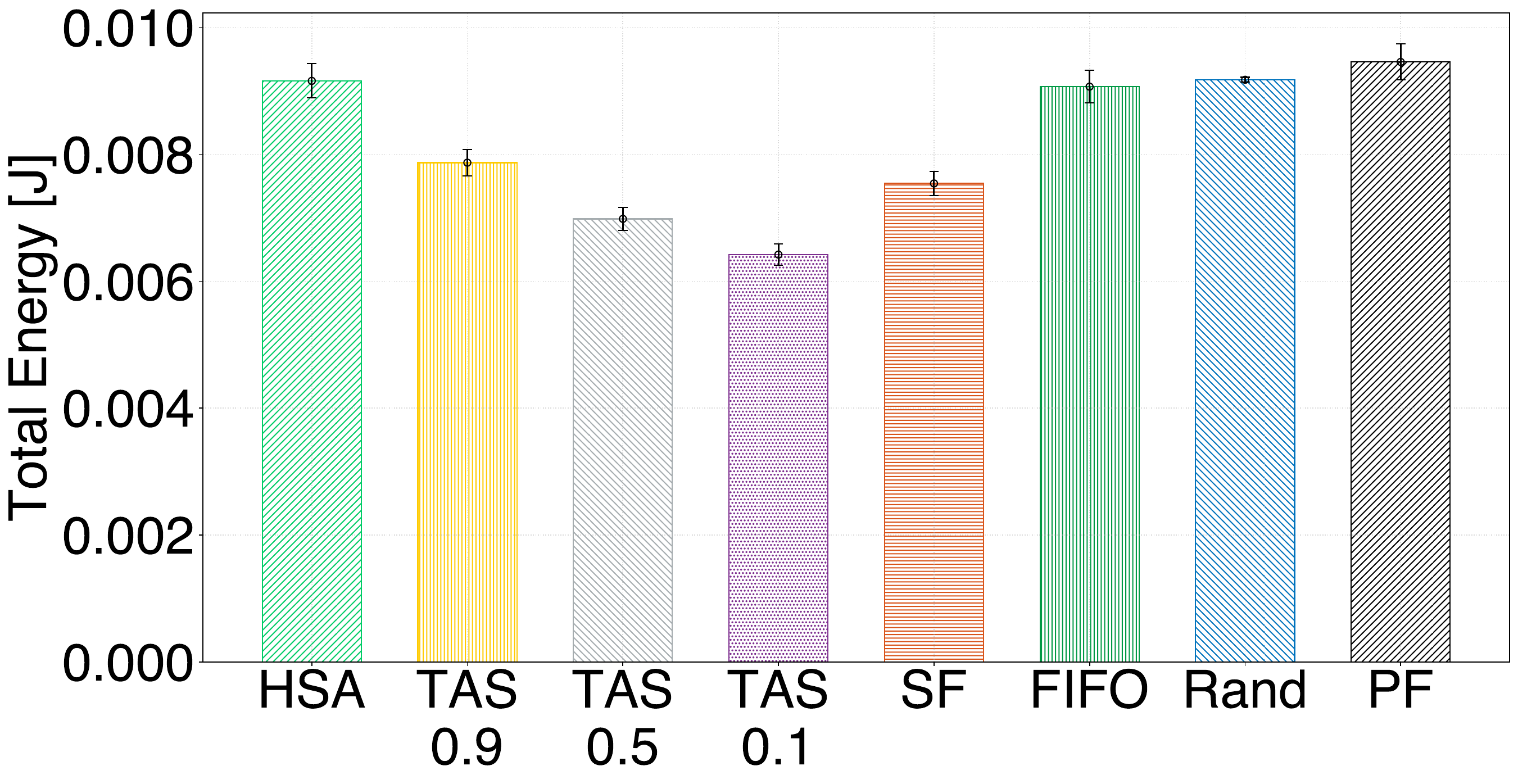}}
   \caption{Mean total energy consumption, averaged over all the STAs that transmitted at least one packet, with 95\% confidence intervals, for TASPER and the different benchmarks, considering progressively complex scenarios, from (a) 16 STAs, to (b) 32 and (c) 64 STAs. From left to right, we compare TASPER (TAS) with \riccardoandfrancesco{HSA}, Optimum (Opt), ShortestFirst (SF), First In First Out (FIFO), Random (Rand) and PriorityFirst (PF) strategies. The numbers under ``Opt'' and ``TAS'' represent the value of $\beta$.}
     \label{fig:energyconsumption}
\end{figure*}

Concerning the mean energy consumption, shown in Fig.~\ref{fig:energyconsumption}, this metric 
has been accurately calculated through ns-3-twt, as mentioned in Sec.~\ref{sub:sim_ns-3-twt}. 
Moreover, we recall that we can obtain both the total energy consumed by the STAs in each PHY layer state 
(as depicted in Fig.~\ref{fig:stateplot}), as well as the total energy consumption of each station, 
during the whole simulation time (shown in Fig.~\ref{fig:energyconsumption}). 
The first important result is evident for all considered numbers of STAs. As expected, the value of $\beta$ 
substantially affects the energy consumption:  higher values of $\beta$ lead to lower mean rejection cost, but at 
the price of a slightly increased energy consumption. Nevertheless, it is worth noting that, 
for both $\beta{=}0.5$ and $\beta{=}0.1$, the total energy consumed by TASPER is still lower than that of the best 
performing benchmark, namely, ShortestFirst. The only exception is for $\beta {=} 0.9$, when TASPER almost completely 
disregards  energy consumption. 
Looking at the case with 16 STAs, moving from $\beta{=}0.1$ to $\beta{=}0.9$ leads to a 20\% energy consumption increase 
(from 14.3\,mJ to 17.4\,mJ), while in the case of 64 STAs it grows by 75\% (from 6.4\,mJ to 11.2\,mJ). Notice, however, 
how the higher the number of STAs, the lower the mean per-STA energy consumption, as each STA will remain in sleep mode 
for a longer time.
Under TASPER, this adds to the fact that, aiming at maximizing the number of scheduled transmissions, 
TASPER tends to select shorter transmissions as the number of STAs increases. It follows that, even when active, 
the time spent by a STA in TX mode becomes shorter. 

\revision{It is worth noting that, although TASPER primarily aims at   maximizing the cumulative value of scheduled transmissions, it does so without violating  traffic deadlines. This can be observed by looking at Table~\ref{tab:miss-perc}, presenting the percentage of deadlines that are missed by the different TWT scheduling algorithms, in a  scenario with 16~STAs and a value $\beta{=}0.9$ for TASPER. Notably, TASPER, Shortest First, and HSA  consistently complete transmissions within their deadlines, thus resulting in a zero packet loss. In contrast, other strategies such as FIFO exhibit a non-negligible percentage of deadline violations.
Also, it is worth remarking that, although Shortest First, HSA and TASPER all ensure that deadlines are met, TASPER outperforms its alternatives in terms of mean rejection cost and energy savings, as shown above.}



\begin{table}[t]
  \centering
    \caption{Deadline miss percentage for TASPER ($\beta=0.9$) and its benchmarks}
  \begin{tabular}{l r}
    \hline
    \textbf{} & \textbf{Deadline miss percentage (\%)} \\
    \hline
    FIFO & 0.52 \\
    Priority First (PF) & 0.16 \\
    Random (Rand) & 0.04 \\
    Shortest First (SF) & 0.00 \\
    HSA & 0.00 \\
    TASPER $\beta=0.9$ (TAS 0.9) & 0.00 \\
    \hline
  \end{tabular}
  \label{tab:miss-perc}
\end{table}

\revision{Moreover, the packet loss rate is zero in all simulations, as the channel is always fully available to the single STA awake in each time slot. Indeed, on the one hand, TASPER computes the transmission duration starting from the packet size and the channel capacity. On the other hand, in the considered scenario the total generated traffic can be scheduled successfully when compared to the channel capacity. This lets TASPER schedule all the transmissions, leading to a null packet loss in the analyzed operational conditions. The latter have been selected to properly assess the capability of TASPER to meet the deadlines when compared to other baselines. As a matter of fact, even the optimal solution would lose packets under different operational conditions that would not allow any feasible scheduling compared to the channel capacity.}

Next, we focus on the trade-off between mean energy consumption and mean rejection cost. We can observe that a 
good trade-off can be achieved for $\beta{=}0.5$, although,  comparing TASPER with  $\beta{=}0.1$ 
(i.e., the case in which TASPER yields the worst performance) to ShortestFirst, TASPER still exhibits a 16\% gain 
with 16 STAs, 14\% with 32 STAs, and 15\%  with 64 STAs. Also, it is important to observe that the reason why 
ShortestFirst consumes less energy than the other benchmarks is because it selects shorter transmission opportunities, 
which,  thanks to shorter \gls{twt} \glspl{sp}, cumulatively require less energy. As for the rejection cost, 
Random remains the worst performing baseline, together with PriorityFirst. For instance, Random consumes 31\% more energy 
than TASPER  with $\beta{=}0.5$ and 64 STAs. In this case,  PriorityFirst does not perform well, as it 
always schedules the TXOPs with the highest priority (thus providing a good mean rejection cost), but 
totally disregarding the fact that long TXOPs lead to high energy consumption. 
\revision{When compared to HSA, TASPER yields a 22\% lower total energy consumption for  $\beta{=}0.5$ and 16 STAs, and 10\% reduction in mean rejection cost in the same scenario.  The main reason for this performance gain lies in the \mbox{energy-aware} design of TASPER, which explicitly addresses the priority–energy trade-off during scheduling, favoring energy-efficient transmissions whenever none of the traffic constraints are violated. In contrast, HSA schedules transmissions based on traffic priority and time deadline alone, without accounting for energy consumption.}
Finally, comparing TASPER to the Optimum again in the case of 16 STAs (Figures~\ref{fig:rejectioncost}(a) 
and~\ref{fig:energyconsumption}(a)), it is evident that TASPER yields an energy consumption just 
slightly higher than  Optimum, while providing an efficient scheduling solution in a much shorter time. 
As an example, for $\beta{=}0.5$, TASPER falls behind Optimum by just 1.4\%. 

Fig.~\ref{fig:concatenation} depicts the performance of the different strategies over a more complicated scheduling scenario, 
including 100 consecutive inter-beacon intervals. In this case, transmissions that are not scheduled in one inter-beacon interval 
are not rejected, rather they  are reconsidered in the next inter-beacon interval, unless their AoI deadline has expired. 
Also, now TASPER is compared only against ShortestFirst, FIFO, and PriorityFirst. Indeed, computing the Optimum in 
a larger case than that with one inter-beacon interval and 16 STAs is impractical, 
while Random is omitted due to the poor performance it exhibited in the previous scenario. 

\begin{figure}[!t]
  \centering
\subfigure[Mean rejection cost with 64 STAs]{
\includegraphics[width=0.3\textwidth]{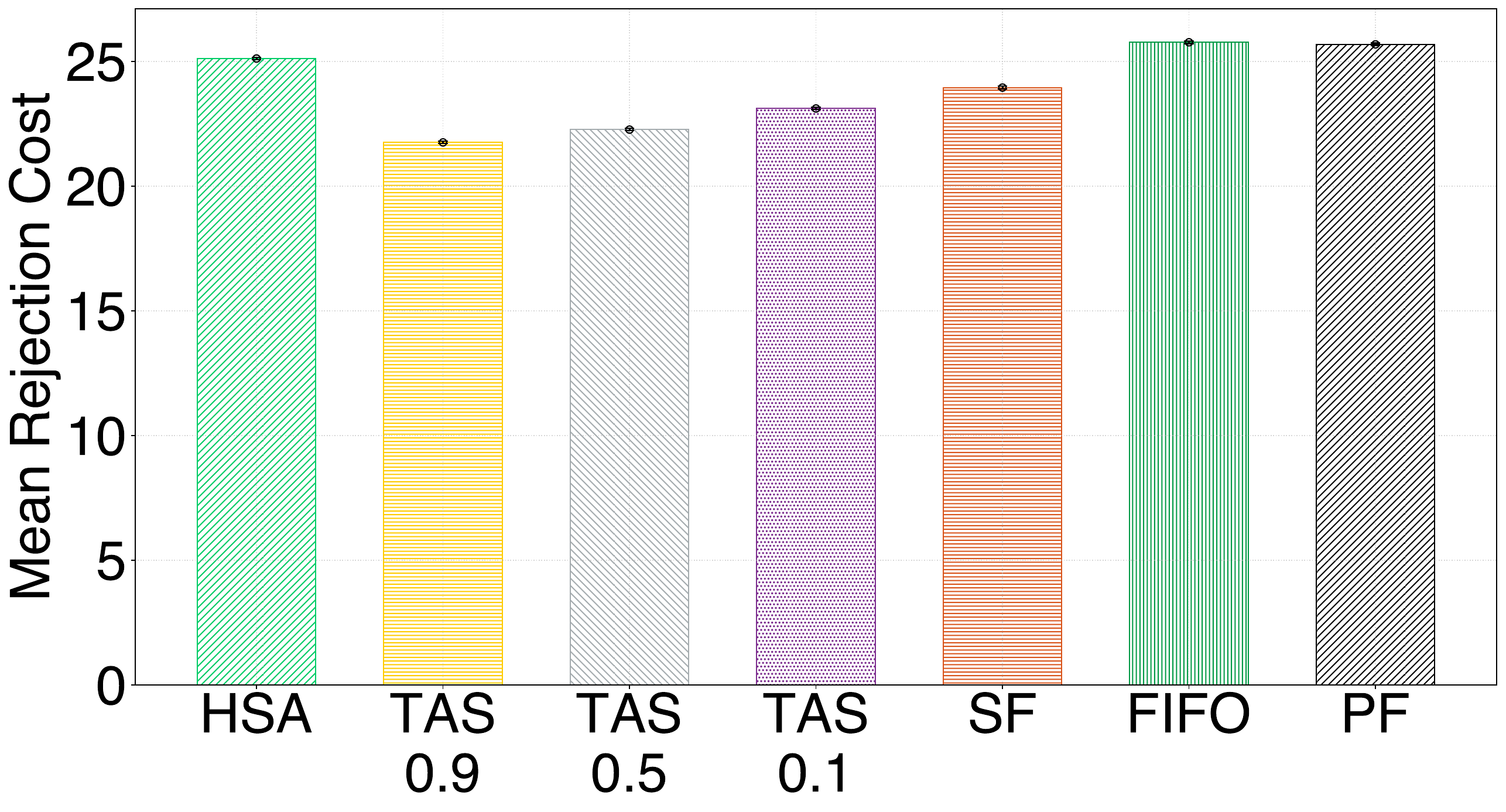}}
\subfigure[Mean total energy consumption with 64 STAs]{
\includegraphics[width=0.3\textwidth]{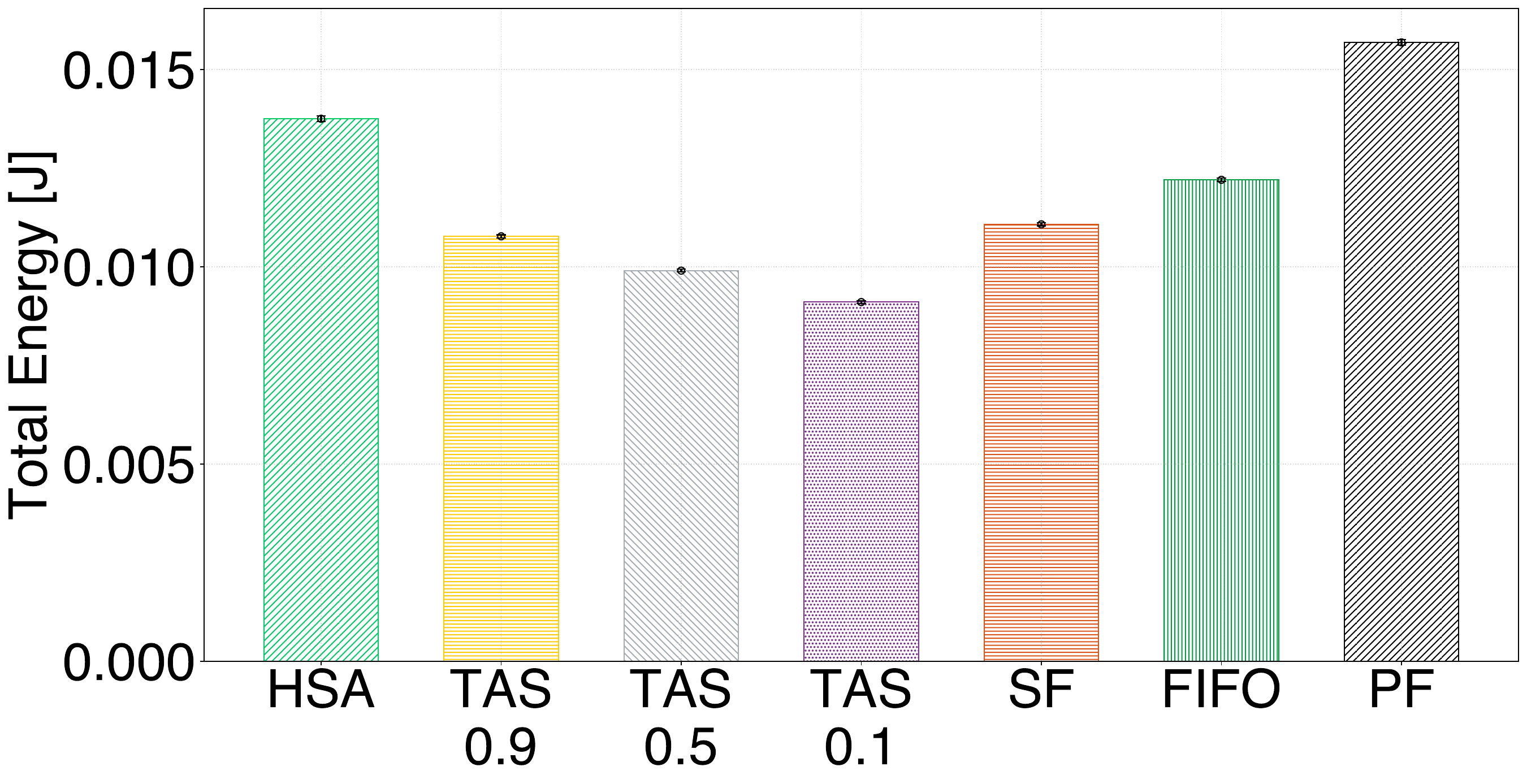}}
   \caption{Mean rejection cost (a) and mean total energy consumption (b), along with 95\% confidence intervals, 
   in the concatenated scenario with 64 STAs. \revision{From left to right, 
    TASPER (TAS) is compared against HSA}, ShortestFirst (SF),  FIFO, and PriorityFirst (PF) strategies. 
    The numbers under ``TAS'' represent the value of $\beta$.}
     \label{fig:concatenation}
\end{figure}

The results in Fig.~\ref{fig:concatenation} further highlight the benefits of the TASPER approach, especially when 
considering energy efficiency and mean rejection cost jointly. 
Looking at Fig.~\ref{fig:concatenation}(a), we notice that the mean rejection cost is significantly higher when compared 
to the performance in the previous, simpler scenario.  
When $\beta{=}0.9$, TASPER shows substantial advantages over its benchmarks, achieving a 16\% (9\%) reduced rejection 
cost compared to FIFO (ShortestFirst). 
Even when considering a smaller value of $\beta$, TASPER still outperforms its alternatives, with a 10\% (4\%) 
reduction in rejection cost compared to FIFO (ShortestFirst), for $\beta{=}0.1$. \revision{TASPER also outperforms HSA with a 13\% reduction in rejection cost, showing how our algorithm can outperform state-of-the-art solutions optimized for priority-based scheduling.}

When focusing on energy consumption, Fig.~\ref{fig:concatenation}(b), TASPER with $\beta{=}0.1$ leads to the highest 
energy savings, achieving a 18\% reduction in energy consumption compared to \mbox{ShortestFirst}, \revision{34\% reduction relative to HSA} and, remarkably, a 42\% reduction 
against PriorityFirst. This confirms the ability of TASPER not only to serve a higher number of transmissions, but also 
to significantly improve the overall energy consumption.
The case for $\beta{=} 0.9$ is even more insightful. The experiments over a single inter-beacon interval showed that, 
when $\beta{=} 0.9$, TASPER consumes more energy than ShortestFirst, as its main objetive is to minimize the mean rejection cost. 
Thus, one may wonder what leads to a different trend in such a more complex scenario. 
The answer lies in the fact that TASPER is able to handle the postponed transmissions very effectively, 
much better than ShortestFirst, ultimately achieving a lower energy consumption even if energy saving is not its primary goal. 

Finally, Fig.~\ref{fig:concatenation} again underlines the importance of selecting a value of $\beta$ 
that represents a good balance between rejection cost and energy consumption. Actually, TASPER with $\beta{=}0.1$ 
yields an energy consumption that is 15\% lower than TASPER with $\beta{=}0.9$ and 8\% lower than TASPER with $\beta{=}0.5$. 
However,  it also exhibits a higher rejection cost (by 6\% w.r.t $\beta{=}0.9$ and 4\% w.r.t. $\beta{=}0.5$).



\section{Experimental Validation}\label{sec:exp}

To validate the observations and the proposed approach in real-world operational conditions, we now 
 show the benefits of our approach 
 in delivering time-critical traffic and in saving energy, using our Wi-Fi IIoT testbed built using commercial TWT-capable STAs and AP. 
In the following, we start  by introducing the experimental testbed we designed  (Sec.~\ref{sec:exp_testbed}). Then,  
we present the configuration employed to first compare  TWT against NoTWT,  and the 
outcomes of such a comparison (Sec.~\ref{sec:exp_twt_setup}). 
Finally, we analyse the performance measured with TASPER versus ShortestFirst, 
which, from the performance evaluation in Sec.~\ref{sec:peva}, resulted to be the 
most performing alternative among the considered benchmarks 
(Sec.~\ref{sec:exp_tasper_setup_res}). 

\subsection{Our Wi-Fi IIoT testbed}\label{sec:exp_testbed}
To represent a real-world IIoT deployment, we selected commercial components and built a testbed 
 comprising 10 STAs and 1 AP, as depicted in Fig.~\ref{fig:testbed}(a). Specifically, the STAs are Espressif ESP32-C6-DevKitC-1, 
 accessible IoT boards that can be purchased for less than 10\,USD. They are based on the ESP32-C6 System on Chip, 
 featuring a 32-bit RISC-V processor, Wi-Fi 6 connectivity (limited to 20-MHz channel in the 2.4\,GHz band, SISO), 
 and support for Individual TWT. These boards can be programmed using the open-source Espressif IoT Development Framework (ESP-IDF), 
 which allows for firmware coding and flashing. ESP-IDF documentation includes an application example for TWT, which has been used 
 to develop the firmware for configuring and running the experiments.
The AP is a Synology WRX560, a TWT-compatible Wi-Fi 6 router, which has been configured to assign static IP addresses to the STAs. 
This simplifies the board configuration since a single firmware can be flashed in all STAs; then, the desired behavior of each 
STA (e.g., the transmission schedule) is inferred based on the assigned IP address. Finally, to collect experimental data, 
an Ubuntu 22.04 host,  acting as the traffic-receiving endpoint, is connected to the router via the 1 GigE LAN interface. 
It integrates an Intel AX200 802.11ax transceiver, which we use to monitor and timestamp frames exchanged in the Wi-Fi network. 
To sniff traffic, we employ the \textit{airmon-ng} and \textit{Wireshark} tools. As the network is secured with WPA2, to access the contents of the transmitted frames, 
which is useful for characterizing the data traffic, the latter is configured with the Pre-Shared Key to decrypt 
the AES-encrypted traffic sent by STAs. Thus, data are obtained by merging the information derived from the received 
application traffic and the sniffed frames. 

\subsection{TWT vs.\ NoTWT: Experimental setup and results}\label{sec:exp_twt_setup}
To demonstrate the benefits of employing TWT to schedule traffic transmissions, we coded 
an application to make the boards transmit one TX of 4,800 bytes each (split in 8 frames of 600 bytes payload) to the Ubuntu 
host once every 10 beacon intervals  (102.4\,ms). To increase network saturation, the allowed 
\gls{mcs} is restricted to MCS0 (PHY rate of 8.6\,Mb/s). To avoid possible overlaps with the AP beacon, 
which could increase transmission delays, all boards generate data 8.192\,ms after the Target Beacon Transmission Time (TBTT).  
This delay is motivated by the AP's behavior: we observed that at each beacon interval, 
the AP issues two consecutive beacons (the second one with a hidden SSID), each requiring up to 3.5\,ms of airtime. 
All boards synchronize their clocks with the AP using the Wi-Fi timing synchronization function (TSF), 
which the Wi-Fi driver exposes to the application generating traffic through ESP-IDF.
To avoid ARP traffic, the Ubuntu host MAC and IP addresses are stored in the boards' firmware.
We tested two configurations:

    \textbullet\, \textbf{NoTWT}, where every board immediately tries to access the channel and transmit traffic 
    as soon as data is generated, as in Wi-Fi networks where TWT is not employed;
    
\textbullet\,  \textbf{TWT}, in which each STA has set up a TWT agreement with the AP for a non-overlapping TWT \gls{sp}. 
    Specifically, in ascending ID order, every STA \gls{sp} starts 9.42\,ms after the previous one. Allowing such a time interval avoids any overlap between subsequent TXOPs,  with each TXOP including also the Block ACK and  
    possible retransmissions, 
    while making all STAs  transmit within a beacon interval.

Each configuration has been tested for more than 12 hours. During this time, we captured more than 35M frames, 
corresponding to more than 48K transmissions by each STA. 


\begin{figure*}
  \centering
\subfigure[Experimental testbed. A 10-port USB hub powers the ESP32 boards. 
        On the screen, a live Wireshark capture shows the sniffed frames. On the right, the Synology AP agrees on the TWT \glspl{sp} 
        and routes the Wi-Fi traffic.]{
\includegraphics[width=0.28\textwidth]{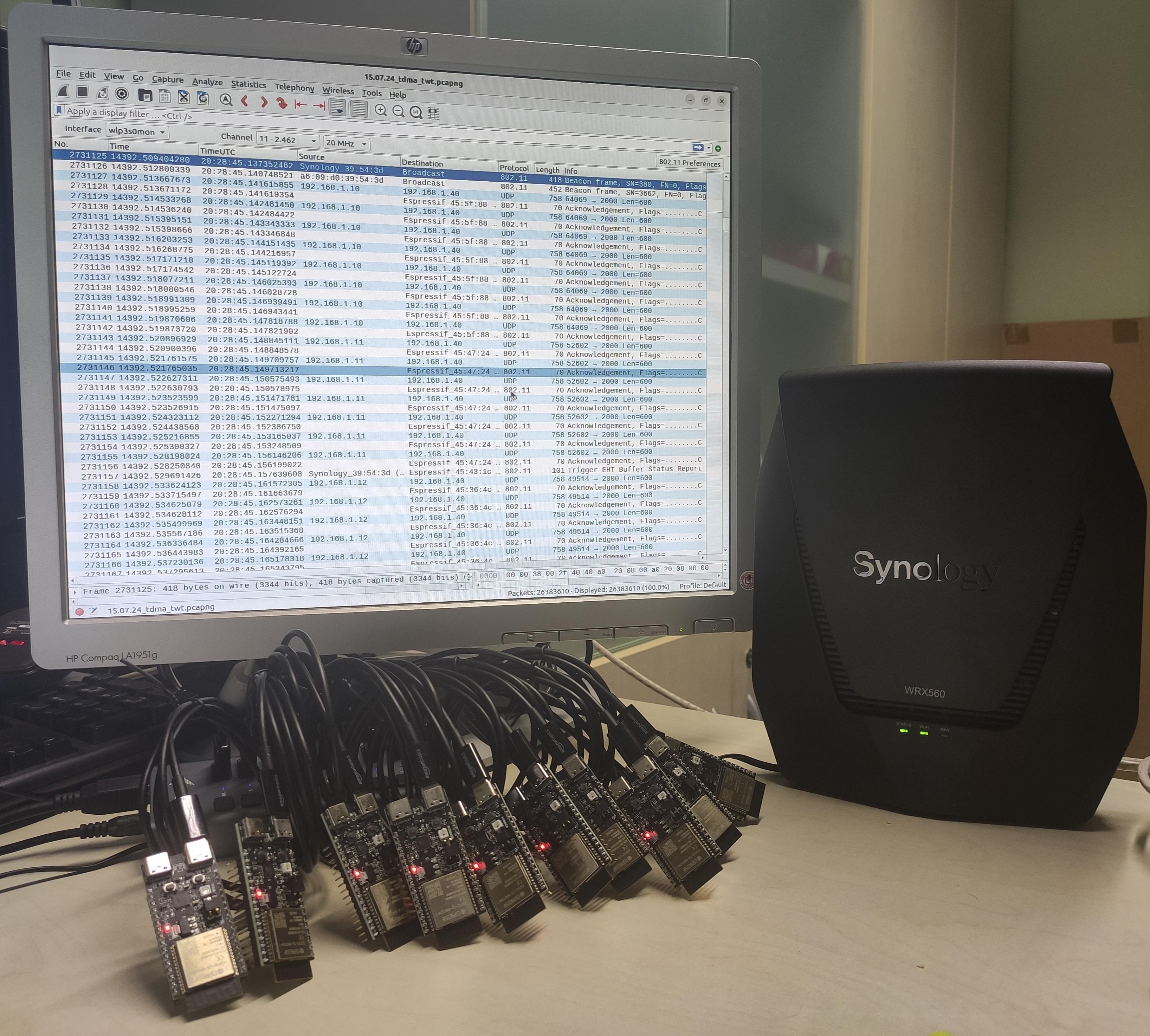}}
\hspace{0.03\textwidth}
\subfigure[NoTWT scenario. The measured AoI is reported for each STA. 
        STAs coordinate through the DCF (i.e., using random backoff), which leads to high AoI.]{
\includegraphics[width=0.28\textwidth]{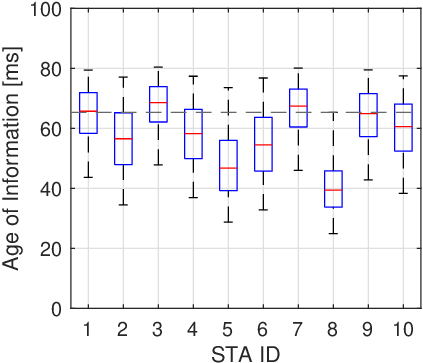}}
\hspace{0.03\textwidth}
\subfigure[TWT scenario. The measured AoI is reported for each STA. STAs have dedicated TWT \glspl{sp} during which they have exclusive availability of the wireless medium.]{
\includegraphics[width=0.28\textwidth]{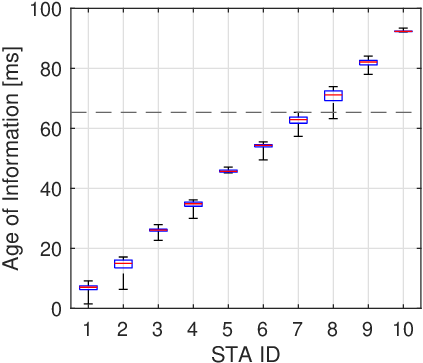}}
   \caption{Experimental testbed and TWT vs.\ NoTWT results. On each box of the box plots, the central mark indicates the median, while the bottom and top edges of the box indicate the 25th and 75th percentiles (resp.). The whiskers extend to the 5th and the 95th percentiles. The horizontal dashed lines identify the AoI threshold not attainable by 95\% of transmission by all STAs in the NoTWT scenario}
   \label{fig:testbed}
\end{figure*}

{\bf AoI measurements.} The results we obtained are presented in the box plots in Figures\,\ref{fig:testbed}(b) and (c), which show the AoI 
measured for each STA in the NoTWT and TWT scenarios, respectively. We recall that the AoI  measures the interval 
between the time of data generation  at the sender and the time of its reception at the receiver. 

When TWT is not used, the AoI varies significantly because of the random backoff:  on average, the median AoI is 58.2\,ms and 
its standard deviation is 12.6\,ms.
More importantly, we observe that none of the STAs is able, for more than 95\% of the times,  to send its traffic 
within 65.3\,ms from the instant when it was generated (horizontal dashed line in Fig\,\ref{fig:testbed} (b)). 
Note that a complete transmission comprises 8 frames, which are not transmitted sequentially. In fact, between 
a frame and the next one, other STAs gain access to the channel, which delays transmission completion for all STAs. 
Thus, these results clearly demonstrate that the NoTWT configuration is unsuitable for time-sensitive networking. 

Conversely, under the TWT mechanism, STAs show consistent AoI  performance, according to their time schedule, 
thanks to the allocation of TXOPs by the AP. On average, the median AoI is 49.1\,ms and its standard deviation is 
2.5\,ms, both lower than those of the NoTWT scenario (16\% and 80\% lower, respectively). 
Here, STAs send all of their frames before another STA attempts to do the same. Because of that, STAs with stricter 
AoI deadlines can be scheduled first, while those with looser AoI deadlines can be left for later. As an example, STA\,1 
completes its transmission within 9.2\,ms more than 95\% of the time. Remarkably, this approach works well 
for a significant number of STAs, as 7 out of the 10 STAs are able to transmit with an AoI consistently 
(i.e., for more than 95\% of the times) lower than 65.3\,ms.  We recall that in NoTWT scenario no STA 
could achieve the same 65.3\,ms AoI for 95\% of the times.  

{\bf Average power consumption measurements.} Not only does TWT improve the AoI performance, but it also allows for 
energy saving as STAs can power off the Wi-Fi transceiver when they are dozing. 
To estimate the energy gain brought by TWT,  we here refer to the average power consumed in the NoTWT and TWT cases. 

First, we estimate the average power consumption in the two cases  
using the datasheet of ESP32-C6-WROOM-1~\cite{esp32-wroom1}, i.e.,  
the module comprising the ESP32-C6 chip, clock oscillator, flash memory, and PCB antenna that are installed in the 
ESP32-C6-DevKitC-1. It reports a current consumption of 251~mA in the transmit mode, 78~mA in the receive mode, and 30~mA 
in the  sleep mode (the energy consumed to change states is not disclosed).  
We assume that in the NoTWT scenario, all STAs generate traffic shortly after the first beacon, then remain active for an entire 
beacon interval (10\% of the total time in the transmit mode), while they access the channel and transmit their traffic. 
After that, they remain active for the following 9 beacon intervals (90\% of the total time in receive  state) while waiting for 
the next traffic generation.
In the TWT scenario, instead, STAs are  in transmit mode only during their assigned TWT \glspl{sp}, which last about 10~ms 
(transmit mode for 1\% of the total time). In the remaining time, they doze to save energy and avoid contending for 
the channel while other STAs are transmitting (99\% of time in modem sleep state). 
Note that a voltage of 5\,V is considered to calculate the power: in the ESP32-C6-DevKitC-1, the 3.3V ESP32-WROOM-1 power supply 
voltage is provided by a voltage linear regulator powered by the USB 5V rail.
In summary, in the traffic scenario described in Sec.~\ref{sec:exp_twt_setup}, we can compute the STAs' average power consumption as: 
    \[P_\text{NoTWT} = 10\% \cdot P^\text{tx} + 90\% \cdot P^\text{rx} = 0.48 \text{W}\]
    \[P_\text{TWT} = 1\% \cdot P^\text{tx} + 99\% \cdot P^\text{sleep} = 0.16 \text{W}\,.\] 
In other words, using TWT, the ESP32 STAs should be able to save 66\% of power compared to 
the NoTWT scenario, consuming 0.16~W instead of 0.48~W on average.

\begin{table}[tb]
\centering
\vspace{0.03in}
\caption{Average testbed power consumption}
\begin{tabular}{|c|r|r|}
\hline
Scenario & Total [W] & Per STA (approx.) [W]\tabularnewline
\hline \hline
NoTWT & 6.343 & 0.47\tabularnewline
\hline
TWT & 3.454 & 0.24\tabularnewline
\hline
\end{tabular}
\label{tab:energy}
\end{table}

We then compare the estimated values of average power consumption with those measured in the testbed, as presented in Tab.~\ref{tab:energy}. We took these measurements with the RCE PM500 power meter connected to the USB hub powering 
the 10 ESP32 boards. In  Tab.~\ref{tab:energy}, we also report the approximate per-STA consumption, computed by subtracting the power loss of the rectifier 
in the power supply of the USB hub (assuming an 80\% efficiency) and the power consumption of the integrated circuit(s) 
inside the hub itself (450~mW).
Under NoTWT, we measure a 6.3~W power consumption, which translates to a STA power consumption of 0.47~W. 
Instead, when TWT is enabled, the total power consumption decreases to 3.4~W, equivalent to a per-STA consumption of 
0.24~W, which is 49\% lower than for the NoTWT case.  
We observe that, in the NoTWT scenario, the estimated power closely matches the measured power. 
Conversely, in the TWT scenario, the estimated consumption is 33\% lower than the measured one. This discrepancy is due   
to factors such as energy consumption while  transitioning from an operational mode to another, which is not accounted 
for in our estimations, and modem wake-ups outside the scheduled SPs for clock re-synchronization. 
Nonetheless, the results confirm that TWT  brings substantial power savings to Wi-Fi\,6 networks, deployed with off-the-shelf 
 devices, supporting time-sensitive applications.

\subsection{TASPER vs.\ ShortestFirst: Experimental setup and results}\label{sec:exp_tasper_setup_res}
We now focus on the experimental comparison between TASPER (with $\beta {=} 1$) and ShortestFirst, 
that is, the best-performing alternative out of those we considered in Sec.~\ref{sub:benchmar}.
The testbed configuration is the same as for the first case study in Sec.~\ref{sec:exp_twt_setup}, 
but we now consider the TASP instance represented in Fig.~\ref{fig:exp_res}(a), which is more representative of a 
real-world scenario where stations have different traffic patterns. This schedule comprises STAs that generate 
periodic traffic with short transmission durations and loose AoI deadlines (STAs $m{=}$1,3,4,6,7,8,9,10) and two STAs 
generating a long transmission, one with a tight AoI deadline (STA $m{=}2$) and the other with a loose one (STA $m{=}5$). 
For each STA, we set a TX priority $p_j {=} 10-m$ and map the TXOP duration $\tau_j$ to a packet size of 1,272 bytes which, 
 using MCS${=}0$, corresponds to a number of frames of $\frac{2}{5}\tau_j$.  

\begin{figure*}
  \centering
\subfigure[Problem instance. White boxes denote the period between TX generations and AoI deadlines; blue bars indicate the TX duration.]{
\includegraphics[width=0.28\textwidth]{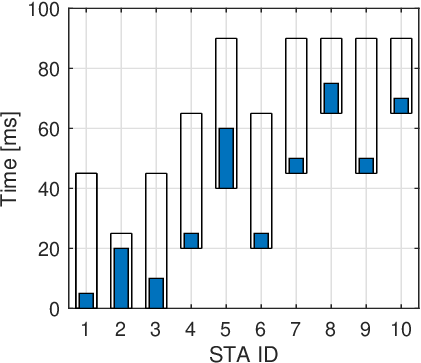}}
\hspace{0.03\textwidth}
\subfigure[TASPER: All STAs are scheduled; for each of them, the measured AoI is reported.]{
\includegraphics[width=0.28\textwidth]{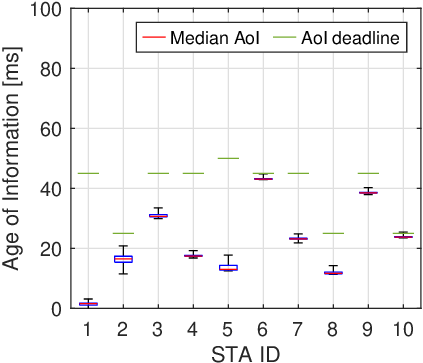}}
\hspace{0.03\textwidth}
\subfigure[ShortestFirst: The AoI is reported for each STA, except for the 2nd, since it cannot be scheduled.]{
\includegraphics[width=0.28\textwidth]{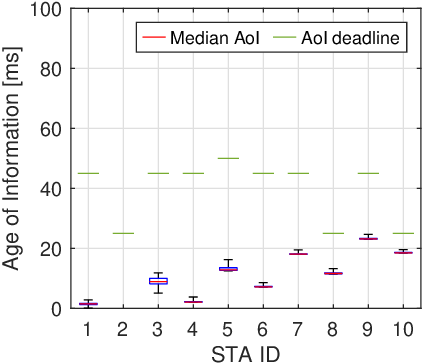}}
   \caption{Problem instance and experimental results: On each box of the box plots, 
   the central mark indicates the median, and the bottom and top edges of the box indicate the 25th and 75th percentiles, 
   respectively. The whiskers extend to the 5th and the 95th percentiles. 
     Light green segments denote the AoI deadlines.}
     \label{fig:exp_res}
\end{figure*}

The results are illustrated in Figures~\ref{fig:exp_res}(b) and (c), which show the AoI measured for each scheduled STA 
with the TASPER and the ShortestFirst strategy, respectively. 
TASPER schedules all STAs in the same order as in the problem instance representation of Fig.~\ref{fig:exp_res}(b), 
with no unused air time between subsequent \glspl{sp}. All the STAs always end their transmissions within the corresponding AoI deadlines, 
except for the 10th STA, whose transmissions exceed the AoI deadline 5.8\% of the time. 
On average, we measure an AoI of 22.05~ms. Interestingly, we observe that on rare occasions the first STA ends its TX in less 
than 1~ms. An imprecise clock synchronization between the Wi-Fi TSF and the board clock, allowing the STA 
to start the transmission earlier, causes this behavior. Indeed, we are not employing triggered-enable TWT to signal STAs 
when their SPs start.  
In contrast, ShortestFirst  first schedules STA 1 and 3,
which have shorter transmission durations, leaving no room to meet the AoI deadline for the traffic of STA~2. 
Since it cannot schedule the TX from STA\,2, it has more time to schedule the requesting STAs,   yielding an  average AoI of 11.66~ms (i.e., 47\% lower than TASPER's). 
However, it fails to schedule  STA\,2, whose service request is thus dropped.  

These results confirm that TASPER successfully leverages TWT to schedule time-sensitive traffic of STAs. 
Notably, using TWT-compatible STAs, broadly available in the market, we demonstrate TASPER's capability of admitting more 
transmissions than alternative approaches such as ShortestFirst by maximizing airtime utilization and avoiding 
AoI deadline violations. 

\section{Related Work}\label{sec:related_work}
Compared to the previous amendment  
(i.e., IEEE 802.11ac), IEEE 802.11ax (a.k.a.  Wi-Fi 6~\cite{802.11ax}), offers new features such as \gls{ofdma}, uplink \gls{mumimo}, and 
\gls{twt}~\cite{802.11ax-tutorial}. \gls{twt},  refined in IEEE 802.11ax but originally proposed in IEEE 802.11ah, 
provides STAs with a new mechanism that allows them to agree on their active and doze times. 
This reduces energy consumption and  network delays, as it decreases medium access 
contention~\cite{twt-tutorial}. Although the standard defines the TWT mechanism, it does not provide any rules or criteria for agreement settings.  
Finding the optimal agreement is thus an open issue, which can be tackled by  formulating a scheduling problem that aims at allocating resources (or machines) to jobs over time periods to optimize one or more objectives. 

\revision{The problem of traffic scheduling is not exclusive to wireless networks such as Wi-Fi; it has also been extensively studied in Ethernet-based networks through the Time-Sensitive Networking (TSN) framework. In particular, the Time-Aware Shaper (TAS) is a TSN mechanism that enables deterministic transmission by controlling gate operations on egress queues of network nodes according to a global time schedule, thus minimizing queuing delays for time-triggered traffic. A comprehensive review of scheduling algorithms designed for TAS is presented in \cite{surveyScheduling}.
}

\revision{Extending TSN concepts to wireless networks introduces additional challenges due to the shared and time-varying nature of the wireless medium. In this context, IEEE 802.15.4 \cite{802.15.4} and LoRa \cite{LoraSemtechOverview} have emerged as two prominent communication technologies enabling low-power wireless networking in industrial and IoT scenarios \cite{chen2009real, chen2014wirelesshart, IndustrialLoRa1, IndustrialLoRa2}. IEEE 802.15.4 provides a low-rate, short-range PHY with support for energy-efficient operation and simple MAC mechanisms such as CSMA/CA. 
LoRa 
enables ultra-long-range communication with extremely low power consumption,  
but at the cost of low data rates and long transmission times, making it suitable for sparse, delay-tolerant applications.
Building on these technologies, a number of industrial protocols and research efforts have aimed to introduce more predictable and time-sensitive behavior. Notably, WirelessHART \cite{iec62591_2016_wirelesshart,chen2014wirelesshart, chen2018joint} extends IEEE 802.15.4  by employing centralized TDMA-based scheduling and frequency hopping to achieve deterministic and reliable communication in mesh topologies, especially in process automation and industrial monitoring. RTLoRa \cite{leonardi2019rt}, instead, modifies LoRaWAN’s purely asynchronous MAC by introducing slotted ALOHA access and lightweight scheduling primitives, with the goal of reducing collisions and supporting soft real-time traffic patterns in dense deployments, which makes it suitable for industrial use cases.  
While both WirelessHART and RTLoRa represent significant advances in bringing determinism and efficiency to low-power wireless networks, they operate under fundamentally different design assumptions compared to Wi-Fi-based solutions. WirelessHART works over a low-rate PHY, relying over fixed time-slot communications, multi-hop mesh routing, and channel hopping. 
RTLoRa, instead, targets long-range communication scenarios, trading off latency and throughput for energy efficiency and coverage, and is constrained by strict duty-cycle regulations and limited bandwidth. In contrast, Wi-Fi networks such as those targeted in this work rely on single-hop communication,
and are designed to support high data rates and highly dynamic traffic patterns.
}

\begin{table*}[t]
\centering
\caption{Comparison of TASPER with related works on wireless industrial networks and TWT scheduling strategies}
\begin{tabular}{|p{1.8cm}|p{3.2cm}|p{2.8cm}|p{1cm}|p{1.4cm}|p{1cm}|p{1cm}|p{2cm}|}
\hline
\textbf{Work (Year) [Ref]} & \textbf{Optimization Objective} & \textbf{Optimization Technique} & \textbf{PHY Layer} & \textbf{Real-Time Guarantee} & \textbf{Energy-Aware} & \textbf{AoI Support} & \textbf{Evaluation Method} \\
\hline
\textbf{TASPER [Ours]} & Maximize traffic acceptance and energy efficiency under AoI constraints & MILP model (TASP) solved via heuristic (TASPER) & IEEE 802.11ax & Soft (AoI-based) & Yes & Yes & ns-3 simulation and real testbed \\
\hline
Chen et al. (2018) \cite{chen2018joint} & Minimize weighted end-to-end delay & Iterative hop-wise scheduling algorithm (HSA) based on MWIS & IEEE 802.15.4 & Soft & No & No & MATLAB simulation \\
\hline
RT-LoRa by Leonardi et al. (2019) \cite{leonardi2019rt} & Provide bounded end-to-end delay & Centralized, superframe-based TDMA strategy & LoRa & Hard & Yes (via QoS classes) & No & OMNeT++/FLoRa simulation \\
\hline
Chen et al. (2021) \cite{9247501} & Maximize throughput under delay constraints & Genetic algorithm & IEEE 802.11ax & Soft (delay-bound) & Yes & No & Simulation with traffic generator \\
\hline
Yang et al. (2021) \cite{twt-scheduling} & Maximize uplink throughput and fairness for TWT STAs & Heuristic grouping: max-rate and proportional-fair algorithms & IEEE 802.11ax & None & No & No & ns-3 simulation (OFDMA uplink) \\
\hline
Schneider et al. (2022) \cite{twt-tsn} & Minimize latency and jitter for TSN traffic & Time-aware scheduling with TWT SP alignment to TAS & IEEE 802.11ax & Soft (latency/jitter bounds) & No & No & Real hardware testbed \\
\hline
Chen et al. (2022) \cite{twt_obss} & Maximize energy efficiency under delay bounds in OBSS & Graph-coloring grouping + deterministic algorithm & IEEE 802.11ax & Soft (delay-bound) & Yes & No & Simulation with traffic generator \\
\hline

Peng et al. (2024) \cite{10652675} & Minimize fixed-duration TWT SPs usage under throughput constraints & Greedy algorithm & IEEE 802.11ax & None & Yes & No & Custom simulator \\
\hline
Dang et al. (2024) \cite{twt-scheduling-w-predictions} & Improve channel efficiency and reduce power consumption under dynamic traffic & Deep learning traffic prediction + heuristic TFST scheduler & IEEE 802.11ax & Soft (service-specific) & Yes & No & Custom simulator with traffic model \\
\hline
\end{tabular}
\label{tab:related_work_comparison}
\end{table*}

\revision{Concerning Wi-Fi based networks, \cite{tsn-wifi} gives an in-depth overview of the IEEE 802.11 standards evolution to support deterministic communication. IEEE 802.1Qbv for time-aware queue gating, IEEE 802.1CB for frame replication and elimination, and AP-driven OFDMA scheduling, can be leveraged to reduce contention and enhance predictability without relying on restrictive \gls{twt} operations.}

\revision{Building upon this, }scheduling strategies based on \gls{twt} have been investigated by several works. Using broadcast TWT, the scheduler in~\cite{9247501}, 
based on a genetic algorithm, takes advantage of OFDMA to avoid channel contention: 
by allocating dedicated resources to each STA, no more than one STA is active at any time. 
Differently from our work, rather than directly allocating \glspl{sp}, \cite{9247501} uses TWT  to wake up a number of STAs that 
is lower or equal to the number of disjoint sets of tones (Resource Units, or RUs), so each STA gets its own RU without colliding with others. 
\revision{Also using OFDMA, \cite{twt-scheduling-w-predictions} proposes a traffic-awareness-based TWT scheduling scheme that leverages spatio-temporal traffic prediction and classification to dynamically adjust TWT parameters and optimize resource allocation across time and frequency. Their scheme, leveraging a greedy scheduling algorithm, improves energy efficiency and QoS by tailoring wake-up intervals and durations to predicted traffic patterns.}
\cite{twt-scheduling}, instead, focuses on throughput and fairness, and proposes two \gls{twt} schedulers: a max-rate scheduler, 
which aims to maximize the overall network throughput, and a proportional fair scheduler, 
which tries to balance network throughput and fairness so that a minimal level of service is guaranteed to all users. 
Aiming at higher energy efficiency and uplink throughput, \cite{twt_obss} presents a power-saving scheme for 
overlapping \gls{bss}. However, such a scheme schedules TWT \glspl{sp} of the same duration, 
thus wasting radio resources in case of short transmissions.
Still considering high-density scenarios with overlapping BSS, \cite{10652675} introduces a TWT-based AP coordination scheme 
to avoid interference. Here, the key idea is to allocate interfering STAs to different TWT \glspl{sp}, and to 
give priority to STAs with high traffic volumes, which leads  
to increased throughput and energy efficiency. 

As for \gls{twt} applied to IIoT scenarios specifically, \cite{twt-tsn}  extends a wired \gls{tsn} network to the wireless domain, 
employing broadcast \gls{twt} to separate traffic flows according to their priority and then different RUs to isolate each flow. However, experimental results show poor performance, as the TWT SPs are emulated through WoWLAN instead of being natively implemented by the used transceivers. 
Looking instead at the performance  on Commercial Off-the-Shelf (COTS) devices, \cite{10279666} investigates the performance of TWT, in the case where 
 TWT agreements are statically configured by STAs (Android smartphones) to decrease energy consumption and are not tailored to traffic characteristics. 
To the best of our knowledge, we are the first to report a performance analysis of TWT on COTS devices where the TWT \glspl{sp} are dynamically scheduled based on the expected traffic patterns. 
\revision{The imperative to maintain data freshness has spurred significant research on \gls{aoi}. Various studies have investigated AoI minimization considering aspects such as maximum \gls{aoi} thresholds and throughput constraints \cite{aoi1}, \cite{aoi2}. However, these AoI-centric formulations typically do not incorporate energy consumption as a primary optimization objective or scheduling constraint, which is a core element of our work where \gls{twt} is central.}

Finally, a preliminary version of our work has appeared in~\cite{puligheddu2024target}, 
where however we introduced a simpler version of the TASP problem and just sketched the TASPER algorithm, 
while showing  its performance obtained via simulation only and under a smaller-scale scenario.

\noindent\textbf{Novelty.}
In summary, the innovative features of our work are two-fold. 
First, concerning energy-saving scheduling, ours is the first work on time-sensitive traffic scheduling 
that addresses energy-saving besides trying to maximize the number of scheduled transmissions.  
    Second, regarding AoI requirements,  existing scheduling problems are designed to comply with 
  delay requirements and do not account for information freshness. 
  Instead, our problem considers AoI constraints to guarantee that the delivered data is not outdated and provides 
  new, fresh information.

\section{Conclusions}
\label{sec:conclusions}
\vspace{-1mm}

In this paper, we presented a novel solution for efficient TWT scheduling in time-critical IIoT scenarios.
First, we provided a set of motivational findings, showing the advantages of TWT in guaranteeing low and deterministic latency as well as in saving energy. 
We then proposed a mathematical model of a TWT-enabled Wi-Fi network and formulated the TWT Acceptance and Scheduling Problem (TASP), proving its  NP-hardness. 
In light of the problem complexity, we envisioned an efficient heuristic algorithm, named TASPER,  and demonstrated its 
effectiveness with respect to baseline strategies. Numerical results, obtained using a realistic IIoT scenario 
and through our ns-3-based TWT simulator,  
show that TASPER achieves a 24.97\% lower mean rejection cost, and up to 14.86\% lower energy consumption than \revision{the best performing baseline.  Additionally, TASPER outperforms HSA, a state-of-the-art solution adapted from WirelessHART \cite{chen2018joint}, reducing the mean rejection cost by up to 26\%, and the energy consumption by up to 34\%.}
Finally, using our IIoT TWT-compatible testbed, we validated TASPER's effectiveness as a traffic scheduling strategy, demonstrating that 
 it admits more transmissions than simpler alternatives without incurring any AoI deadline violations.

Future work will  extend TASPER to an OFDMA scenario, where Wi-Fi transmissions are allocated in Resource Units 
spanning both the time and frequency dimensions, and  further investigate the stations' energy consumption.

\bibliographystyle{IEEEtran} 
\bibliography{main}

\vspace{-1cm}

\begin{IEEEbiography}
    {Fabio Busacca} is an assistant professor at the University of Catania, Italy. His main research interests are LPWAN protocols for the IoT, AI applied to next-generation communication networks, and underwater networks.
\end{IEEEbiography}

\vspace{-1.2cm}

\begin{IEEEbiography}
    {Corrado Puligheddu} is an assistant professor at Politecnico di Torino, Italy. His main area of interest is the application of machine learning to wireless networks, focusing on radio resource management and network orchestration.
\end{IEEEbiography}

\vspace{-1.2cm}

\begin{IEEEbiography}
    {Francesco Raviglione} is an assistant professor at Politecnico di Torino, Italy. His main areas of interest are wireless and vehicular networks.
\end{IEEEbiography}

\vspace{-1.2cm}

\begin{IEEEbiography}
    {Riccardo Rusca} is a research fellow at Politecnico di Torino, Italy.
    His main areas of interest are crowd monitoring and time sensitive networking.
\end{IEEEbiography}

\vspace{-1.2cm}

\begin{IEEEbiography}
    {Claudio Casetti} is a Full Professor with Politecnico di Torino, Italy. His research interests are vehicular networks, ITS, 5G/6G, and IoT systems.
\end{IEEEbiography}

\vspace{-1.2cm}

\begin{IEEEbiography}
    {Carla Fabiana Chiasserini} is Full Professor with  Politecnico di Torino, Italy. Her research interests are in the design, modeling, and performance evaluation of mobile networks and services.
\end{IEEEbiography}

\vspace{-1.2cm}

\begin{IEEEbiography}
    {Sergio Palazzo} is a Full Professor with the  Università di Catania, Italy. His research interests include mobile systems, wireless and satellite networks, and traffic engineering.
\end{IEEEbiography}

\end{document}